%% file: Main.tex
\documentclass[11pt]{article}
\usepackage{fullpage} 
\usepackage{xspace}
\usepackage{comment}

\usepackage{amsmath}
\usepackage{amsthm}
\usepackage{amsfonts}
\usepackage{amssymb}
\usepackage{algorithm}
\usepackage{algpseudocode}

\usepackage{thmtools}
\usepackage{thm-restate}
\usepackage{subcaption}

\usepackage{hyperref}
\usepackage[dvipsnames]{xcolor}
\definecolor{DarkGreen}{rgb}{0.1333,0.5451,0.1333}
\definecolor{DarkRed}{rgb}{0.8,0,0}
\definecolor{DarkBlue}{rgb}{0,0,0.44}
\hypersetup{
  colorlinks=true,
  citecolor=blue,
  linkcolor=blue,
  urlcolor=blue,
}

\newtheorem{theorem}{Theorem}[section]
\newtheorem{lemma}[theorem]{Lemma}
\newtheorem{corollary}[theorem]{Corollary}

\theoremstyle{definition}
\newtheorem{definition}[theorem]{Definition}

\def\defeq{\stackrel{\mathrm{def}}{=}}

\usepackage{fancybox}

\usepackage{pgf,tikz}
\usepackage{mathrsfs}
\usetikzlibrary{arrows}

\newenvironment{fminipage}%
  {\begin{Sbox}\begin{minipage}}%
  {\end{minipage}\end{Sbox}\fbox{\TheSbox}}

\newenvironment{algbox}[0]{\vskip 0.2in
\noindent 
\begin{fminipage}{6.3in}
}{
\end{fminipage}
\vskip 0.2in
}

\usepackage{color}

\newcommand{\Rmin}[0]{R_{\mathrm{min}}}
\newcommand{\minimizer}{\textsc{Minimizer}}

\DeclareMathOperator*{\argmin}{arg\,min}
\newcommand{\OV}{\textsc{OrthogonalVectors}\xspace}
\newcommand{\Vvec}{V_{\normalfont\text{vec}}\xspace}
\newcommand{\Vdim}{V_{\normalfont\text{dim}}\xspace}
\newcommand{\Vpad}{V_{\normalfont\text{pad}}\xspace}
\newcommand{\vvec}{v_{\normalfont\text{vec}}\xspace}
\newcommand{\vdim}{v_{\normalfont\text{dim}}\xspace}
\newcommand{\vpad}{v_{\normalfont\text{pad}}\xspace}

\newcommand{\Gfill}[1]{G^{+}_{#1}}
\newcommand{\Vfill}[1]{V^{+}_{#1}}
\newcommand{\Efill}[1]{E^{+}_{#1}}
\newcommand{\degfill}[1]{\deg^{+}_{#1}}
\newcommand{\Nfill}[1]{N^{+}_{#1}}

\newcommand{\Gcomp}[1]{G^{\circ}_{#1}}
\newcommand{\Vcomp}[1]{V^{\circ}_{\normalfont\text{comp}#1}}
\newcommand{\Ecomp}[1]{E^{\circ}_{#1}}
\newcommand{\Ncomp}[1]{N^{\circ}_{\normalfont\text{comp}#1}}
\newcommand{\degcomp}[1]{\deg^{\circ}_{\normalfont\text{comp}#1}}

\newcommand{\Vrem}[1]{V^{\circ}_{\normalfont\text{rem}#1}}
\newcommand{\Nrem}[1]{N^{\circ}_{\normalfont\text{rem}#1}}
\newcommand{\degrem}[1]{\deg^{\circ}_{\normalfont\text{rem}#1}}

\newcommand{\MinDeg}{\textsc{MinDegreeOrdering}\xspace}

\def\prob#1#2{\mbox{\textnormal{Pr}}_{#1}\left[ #2 \right]}

\renewcommand\aa{\boldsymbol{\mathit{a}}}

\newcommand{\variable}[1]{{\normalfont\texttt{#1}}}

\newcommand{\Exp}{\normalfont\text{Exp}}
\newcommand{\nnz}{\normalfont\text{nnz}}
\newcommand{\E}{\mathbb{E}}

\begin{document}

\title{Graph Sketching Against Adaptive Adversaries\\ Applied to the Minimum Degree Algorithm}

\author{
Matthew Fahrbach\thanks{Supported in part by a National Science
Foundation Graduate Research Fellowship under grant DGE-1650044.}
\\
Georgia Tech\\
\href{mailto:matthew.fahrbach@gatech.edu}{\texttt{matthew.fahrbach@gatech.edu}}
\and
Gary L. Miller\thanks{This material is based on work supported by the
National Science Foundation under Grant No. 1637523.}\\
CMU\\
\href{mailto:glmiller@cs.cmu.edu}{\texttt{glmiller@cs.cmu.edu}}
\and
Richard Peng\thanks{This material is based on work supported by the
National Science Foundation under Grant No. 1637566.}\\
Georgia Tech\\
\href{mailto:rpeng@cc.gatech.edu}{\texttt{rpeng@cc.gatech.edu}}
\and
Saurabh Sawlani\footnotemark[3]\\
Georgia Tech\\
\href{mailto:sawlani@gatech.edu}{\texttt{~~~~~sawlani@gatech.edu~~~~}}
\and
Junxing Wang\\
CMU\\
\href{mailto:junxingw@cs.cmu.edu}{\texttt{junxingw@cs.cmu.edu}}
\and
Shen Chen Xu\footnotemark[2]\\
Facebook\thanks{Part of this work was done while at CMU.}\\
\href{mailto:shenchex@cs.cmu.edu}{\texttt{shenchex@cs.cmu.edu}}
}
\maketitle

\input{abstract}

\pagenumbering{gobble}

\vfill

\pagebreak

\pagenumbering{arabic}

\input{Introduction}

\input{Preliminaries}

\input{Overview}

\input{Sketching}

\input{Decorrelation}

\input{DegreeEstimation}

\input{DynamicGraphs}

\input{Hardness}

\section*{Acknowledgements}

We thank John Gilbert and Gramoz Goranci for many helpful discussions regarding
various topics in this paper.  We also would like to acknowledge Animesh
Fatehpuria for independently obtaining the construction of covering set systems
in Lemma~\ref{lem:CoveringSetSystem}.

\bibliographystyle{alpha}
\bibliography{references}

\begin{appendix}
	
\input{SketchingProofs.tex}

\end{appendix}

\end{document}

%% file: abstract.tex
\begin{abstract}

Motivated by the study of matrix elimination orderings in combinatorial
scientific computing, we utilize graph sketching and local sampling to give a
data structure that provides access to approximate fill degrees of a matrix
undergoing elimination in $O(\text{polylog}(n))$ time per elimination and query.
We then study the problem of using this data structure in the minimum
degree algorithm, which is a widely-used heuristic for producing elimination
orderings for sparse matrices by repeatedly eliminating the vertex with
(approximate) minimum fill degree.
This leads to a nearly-linear time algorithm for generating
approximate greedy minimum degree orderings.
Despite extensive studies of algorithms for elimination orderings 
in combinatorial scientific computing, our result is the first rigorous
incorporation of randomized tools in this setting, as well as the first
nearly-linear time algorithm for producing elimination orderings
with provable approximation guarantees.

While our sketching data structure readily works in the oblivious
adversary model, by repeatedly querying and greedily updating itself,
it enters the adaptive adversarial model where the underlying sketches
become prone to failure due to dependency issues with their internal randomness.
We show how to use an additional sampling procedure to circumvent this
problem and to create an independent access sequence.
Our technique for decorrelating the interleaved queries and updates to this
randomized data structure may be of independent interest.
\end{abstract}

%% file: Introduction.tex
\section{Introduction}
\label{sec:Introduction}

Randomization has played an increasingly fundamental role in the
design of modern data structures.
The current best algorithms for
fully-dynamic graph connectivity~\cite{KapronKM13,NanongkaiSW17,
NanongkaiS17,Wulffnilsen17},
shortest paths~\cite{HenzingerKN14,AbrahamCK17},
graph spanners~\cite{BaswanaKS12},
maximal matchings~\cite{BaswanaGS15,Solomon16},
and the dimensionality-reductions of large
matrices~\cite{Woodruff14,CohenMP16,KapralovLMMS17,KyngPPS17} all critically
rely on randomization.
An increasing majority of these data structures operate under the
\emph{oblivious adversary model}, which assumes that updates are generated
independently of the internal randomness used in the data structure.
In contrast,
many applications of data structures are \emph{adaptive}---meaning
that subsequent updates may depend on the output of previous queries.
A classical example of this paradigm 
is the combination of greedy algorithms with data structures,
including Dijkstra's algorithm for computing shortest paths
and Kruskal's algorithm for finding minimum spanning trees.
The limitations imposed by adaptive adversaries are beginning to 
receive attention in the dynamic connectivity~\cite{NanongkaiS17,NanongkaiSW17}
and spanner~\cite{BodwinK16} literature,
but even for these problems there remains a substantial gap
between algorithms that work in the adaptive adversary model 
and those that work only against oblivious adversaries~\cite{BaswanaKS12,KapronKM13}.

Motivated by a practically important example of adaptive
invocations to data structures for greedy algorithms, we study the minimum degree
algorithm for sparse matrix factorization and linear system
solving~\cite{David16:survey}.
This heuristic for precomputing an efficient pivot ordering is ubiquitous in
numerical linear algebra libraries that handle large sparse matrices~\cite{Matlab17},
and relies on a graph-theoretic interpretation of Gaussian elimination.
In particular, the variables and nonzeros in a linear system correspond to
vertices and edges in a graph, respectively. When the variable associated with
vertex $u$ is eliminated, a clique is induced on the neighborhood of $u$,
and then $u$ is deleted from the graph.
This heuristic repeatedly eliminates the vertex of minimum
degree in this graph, which corresponds to the variable with the
fewest nonzeros in its row and column.

Computing elimination orderings that minimize the number of
additional nonzeros, known as \emph{fill}, has been shown to be computationally
hard~\cite{Yannakakis81,NatanzonSS00}, even in parameterized
settings~\cite{KaplanST99,FominV13,WuAPL14,Bliznets16,CaoS17}.  However, the
practical performance of direct methods has greatly benefited from more
efficient algorithms for analyzing elimination
orderings~\cite{Amestoy2004,DGLN04}.
Tools such as elimination trees~\cite{Liu90,GilbertNP94} 
can implicitly represent fill in time that is nearly-linear in the
number of \emph{original} nonzeros,
which allows for efficient prediction and reorganization of
future computation and, more importantly, memory bandwidth.
In contrast to the abundance of algorithms built on examining elimination
orderings via implicit representation~\cite{HendricksonP06,NaumannS12:book},
surprisingly little attention has been given to producing
elimination orderings implicitly.
In the survey by Heggernes et al.~\cite{HeggernesEKP01}, the authors
give an $O(n^{2}m)$ algorithm for computing a minimum degree ordering,
which is more than the cost of Gaussian elimination itself
and significantly more than the nearly-linear time algorithms
for analyzing such orderings~\cite{GilbertNP94}.

\paragraph*{Main Results.}
We begin our study by combining implicit representations of
fill with graph sketching.
The nonzero entries of a partially eliminated matrix can be represented
as the set of vertices reachable within two hops in a graph
that undergoes edge contractions~\cite{GilbertNP94}.
This allows us to incorporate $\ell_0$-sketches~\cite{Cohen97}, which
were originally developed to estimate the cardinality
of reachable sets of vertices in directed graphs.
By augmenting $\ell_0$-sketches with suitable data structures,
we obtain the following result for dynamically maintaining fill structure.

\begin{theorem}
\label{thm:SketchMain}
Against an oblivious adversary,
we can maintain $(1 \pm \epsilon)$-approximations to the degrees of
the graph representation of a matrix undergoing elimination
in $O(\log^3{n} \epsilon^{-2})$ per operation.
\end{theorem}

We also give an exact version of this data structure for cases
where the minimum degree is always small
(e.g., empirical performance of Gaussian elimination on grid graphs~\cite{bornstein1997parallelizing}).
Ignoring issues of potential dependent randomness,
the approximation guarantees of this data structure provide us with an
ordering that we call an \emph{approximate greedy minimum degree ordering},
where at each step a vertex whose degree is close to the minimum is pivoted.
It is unclear if such an ordering approximates a true minimum degree
ordering, but such guarantees are more quantifiable than previous
heuristics for approximating minimum degree orderings~\cite{AmestoyDD96,HeggernesEKP01}.

However, using this randomized data structure in a greedy manner exposes the
severe limitations of data structures that only work in the oblivious adversary
model.
The updates (i.e.\ the vertices we eliminate) depend on the output to previous
minimum-degree queries, and hence its own internal randomness.  The main result
in this paper is an algorithm that uses dynamic sketching,
as well as an additional routine for estimating degrees via local
sampling, to generate an approximate greedy minimum degree sequence in
nearly-linear time against adaptive adversaries.

\begin{theorem}
\label{thm:main}
Given an $n \times n$ matrix $A$ with nonzero graph structure $G$
containing $m$ nonzeros, we can produce a $(1+\epsilon)$-approximate greedy
minimum degree ordering in $O(m \log^{5}n \epsilon^{-2})$ time.
\end{theorem}

\paragraph{Techniques.}
Several components of our algorithm are highly tailored to the minimum
degree algorithm. For example, our dynamic sketches and local degree
estimation routine depend on the implicit representation of intermediate states
of Gaussian elimination~\cite{GilbertNP94}.
That said, our underlying randomized techniques
(e.g., $\ell_0$-sketches~\cite{Cohen97} and
wedge sampling~\cite{KallaugherP17,EdenLRS17})
are new additions to combinatorial scientific computing.

The primary focus of this paper is modifying the guarantees in the oblivious
adversary model from Theorem~\ref{thm:SketchMain} to
work within a greedy loop (i.e.\ an adaptive adversary)
to give Theorem~\ref{thm:main}.
However, we do not accomplish this by making the queries
deterministic or worst-case as in~\cite{BodwinK16,NanongkaiS17,NanongkaiSW17}.
Instead, we use an external randomized routine for estimating fill degrees
to create a fixed sequence of updates.
The randomness within the sketching data structure then becomes independent
to the update sequence,
but its internal state is still highly useful for determining which vertices
could have approximate minimum degree.
We then efficiently construct the update sequence using recent developments
for randomized graph algorithms that use
exponential random variables~\cite{MillerPX13,MillerPVX15}.
Our use of sketching can also be viewed as a pseudodeterminstic algorithm
whose goal is to efficiently recover a particular sequence of
vertices~\cite{GatG11,GoldreichGR13}.
We believe that both of these views are valuable to the study of randomness and
for better understanding the relationship between oblivious and adaptive
adversaries.

\paragraph{Organization.}
In Section~\ref{sec:Preliminaries} we formalize the implicit representation of
fill and variants of minimum degree orderings.
In Section~\ref{sec:Overview} we give an overview of our
results, along with a brief description of the algorithms
and techniques we employ.
The use of sketching and sampling to obtain our exact and approximate
algorithms are given in Section~\ref{sec:Sketching} and Section~\ref{sec:Decorrelation},
respectively.
We also detail our derandomization routine in Section~\ref{sec:Decorrelation},
which is crucial for using our randomized data structure against an adaptive
adversary.
In Section~\ref{sec:DegreeEstimation} we demonstrate how to estimate fill degrees
via local sampling, and in Section~\ref{sec:DynamicGraphs} we show how to
maintain sketches as vertices are pivoted.
Lastly, in Section~\ref{sec:Hardness} we discuss hardness results
for computing the minimum degree of a vertex in a partially eliminated
system and also for producing a minimum degree ordering.

%% file: Preliminaries.tex
\section{Preliminaries}
\label{sec:Preliminaries}

We assume that function arguments are pointers to
objects instead of the objects themselves, and thus
passing an object of size $O(n)$ does not cost $O(n)$ time and space.
This is essentially the ``pass by reference'' construct in
high-level programming languages.

\subsection{Gaussian Elimination and Fill Graphs}
\label{subsec:PreliminariesGaussian}

Gaussian elimination is the process of repeatedly
eliminating variables from a system of linear equations,
while maintaining an equivalent system on the remaining variables.
Algebraically, this involves taking an equation involving
a target variable and subtracting (a scaled version of) this equation
from all others involving the target variable.
We assume throughout the paper that the systems are symmetric
positive definite (SPD) and thus the diagonal will remain
positive, allowing for any pivot order.
This further implies that we can apply elimination operations to
columns in order to isolate the target variable, resulting in the Schur
complement.

A particularly interesting fact about Gaussian elimination is that the
numerical Schur complement is unique irrespective of the pivoting order.
Under the now standard assumption that nonzero elements do not cancel each
other out~\cite{GeorgeL89}, this commutative property also holds for the
combinatorial nonzero structure. By interpreting the nonzero structure of a symmetric
matrix~$A$ as an adjacency matrix for a graph $G$, we can define
the change to the nonzero structure of $A$ as a graph-theoretic operation
on $G$ analogous to the Schur complement.

Our notation extends that of Gilbert, Ng, and Peyton~\cite{GilbertNP94}, who
worked with known elimination orderings and treated the entire fill
pattern (i.e.\ additional nonzeros entries) statically.
Because we work with partially eliminated states, we will need to distinguish
between the \emph{eliminated} and \emph{remaining} vertices in $G$. We
implicitly address this by letting $x$ and $y$ denote eliminated vertices
and by letting $u$, $v$, and $w$ denote remaining vertices.
The following definition of a fill graph allows us to determine the nonzero
structure on the remaining variables of a partially eliminated system.

\begin{definition}
The \emph{fill graph} $\Gfill{} = (\Vfill{},\Efill{})$
is a graph on the remaining vertices such that the edge
$(u,v) \in \Efill{}$
if $u$ and $v$ are connected by a (possibly empty) path of eliminated vertices.
\end{definition}

This characterization of fill means that we can readily compute
the \emph{fill degree} of a vertex~$v$, denoted by
$\degfill{}(v)=|\Nfill{}(v)|$, in a partially eliminated state without
explicitly constructing the matrix.
We can also iteratively form $\Gfill{}$ from the original graph $G$ by
repeatedly removing an eliminated vertex $x$ along with its incident edges, and
then adding edges between all of the neighbors of~$x$ to form a clique.
This operation gives the nonzero structure of the Schur complement.

\begin{lemma}
\label{lem:ComputeFill}
For any graph $G = (V,E)$ and vertex $v \in V$, given an elimination ordering
$S$ we can compute $\degfill{}(v)$ at the step when $v$ is eliminated in
$O(m)$ time.
\end{lemma}

\begin{proof}
Mark all the vertices appearing in $S$ before $v$ as eliminated, and mark the
rest as remaining.
Run a breadth-first search from $v$ that terminates at remaining vertices (not
including $v$).
Let $T$ be the set of vertices where the search terminated.
By the definition of $\Gfill{}$ we have $\degfill{}(v) = |T|$.
\end{proof}

This kind of path finding among eliminated vertices adds an additional layer of
complexity to our data structures.  To overcome this, we contract eliminated
vertices into their connected components (with respect to their induced subgraph
in $G$), which leads to the component graph.

\begin{definition}
We use $\Gcomp{} = (\Vcomp{}, \Vrem{}, \Ecomp{})$ to denote the
\emph{component graph}.
The set of vertices in $\Vcomp{}$ is formed by contracting edges between
eliminated vertices, and the set of vertices that have not been
eliminated is $\Vrem{}$. The set of edges $\Ecomp{}$ is implicitly given
by the contractions.
\end{definition}

\noindent
Note that $\Gcomp{}$ is quasi-bipartite, as the contraction rule
implies there are no edges between vertices in $\Vcomp{}$.
It will be useful to refer to two different kinds of
neighborhoods in a component graph. For any vertex $v$ in $\Gcomp{}$,
let $\Nrem{}(v)$ be the set of neighbors of $v$ are in~$\Vrem{}$,
and let
$\Ncomp{}(v)$ denote the neighbors of $v$ that are in $\Vcomp{}$.
Analogously, we use the notation $\degrem{}(v) = |\Nrem{}(v)|$
and $\degcomp{}(v) = |\Ncomp{}(v)|$.

\vspace{0.25cm}
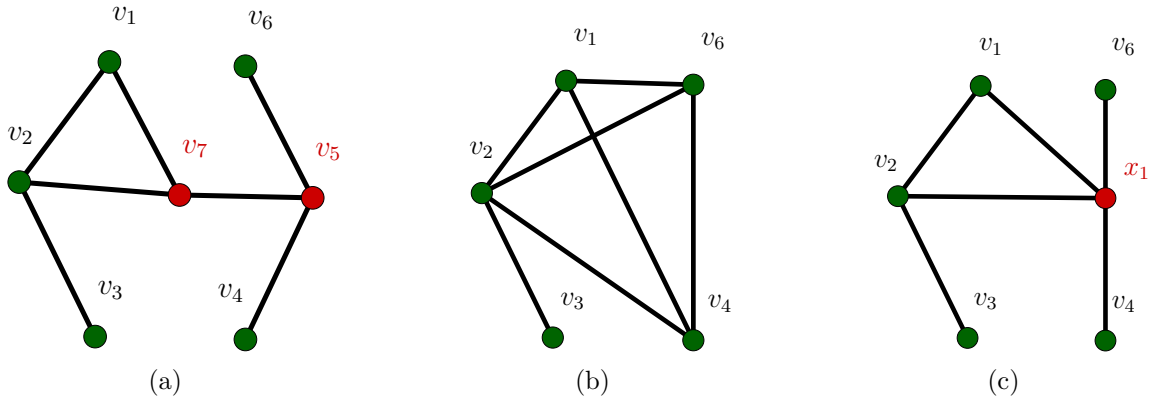
\begin{figure}[H]
  \centering
  \begin{subfigure}[b]{0.31\linewidth}
      \centering
      \resizebox{\linewidth}{!}{
      \centering
        \definecolor{ccqqqq}{rgb}{0.8,0,0}
        \definecolor{qqwuqq}{rgb}{0,0.39215686274509803,0}
        \begin{tikzpicture}[scale=1,line cap=round,line join=round,>=triangle 45,x=1cm,y=1cm]
        \draw [line width=2pt] (-13.86,2.67)-- (-12.6,4.35);
        \draw [line width=2pt] (-13.86,2.67)-- (-12.8,0.51);
        \draw [line width=2pt] (-13.86,2.67)-- (-11.62,2.49);
        \draw [line width=2pt] (-11.62,2.49)-- (-12.6,4.35);
        \draw [line width=2pt] (-11.62,2.49)-- (-9.76,2.45);
        \draw [line width=2pt] (-9.76,2.45)-- (-10.7,4.29);
        \draw [line width=2pt] (-10.7,0.47)-- (-9.76,2.45);
        \begin{scriptsize}
        \draw [fill=qqwuqq] (-12.6,4.35) circle (4.5pt);
        \draw (-12.38,5) node {\large $v_1$};
        \draw [fill=qqwuqq] (-13.86,2.67) circle (4.5pt);
        \draw (-13.84,3.32) node {\large $v_2$};
        \draw [fill=ccqqqq] (-9.76,2.45) circle (4.5pt);
        \draw[color=ccqqqq] (-9.54,3.1) node {\large $v_5$};
        \draw [fill=qqwuqq] (-10.7,0.47) circle (4.5pt);
        \draw (-10.90,1.12) node {\large $v_4$};
        \draw [fill=qqwuqq] (-12.8,0.51) circle (4.5pt);
        \draw (-12.58,1.16) node {\large $v_3$};
        \draw [fill=qqwuqq] (-10.7,4.29) circle (4.5pt);
        \draw (-10.48,4.94) node {\large $v_6$};
        \draw [fill=ccqqqq] (-11.62,2.49) circle (4.5pt);
        \draw[color=ccqqqq] (-11.4,3.14) node {\large $v_7$};
        \end{scriptsize}
        \end{tikzpicture}
    }
    \caption{}
  \end{subfigure}
  \hspace{0.8cm}
  \begin{subfigure}[b]{0.25\linewidth}
  \centering
  \resizebox{\linewidth}{!}{
    \centering
  
	\definecolor{ccqqqq}{rgb}{0.8,0,0}
	\definecolor{qqwuqq}{rgb}{0,0.39215686274509803,0}
	\begin{tikzpicture}[scale=1,line cap=round,line join=round,>=triangle 45,x=1cm,y=1cm]
	\draw [line width=2pt] (-7.6,2.67)-- (-6.34,4.35);
	\draw [line width=2pt] (-7.6,2.67)-- (-6.54,0.51);
	\draw [line width=2pt] (-6.34,4.35)-- (-4.44,0.47);
	\draw [line width=2pt] (-4.44,0.47)-- (-4.44,4.29);
	\draw [line width=2pt] (-4.44,4.29)-- (-6.34,4.35);
	\draw [line width=2pt] (-7.6,2.67)-- (-4.44,0.47);
	\draw [line width=2pt] (-7.6,2.67)-- (-4.44,4.29);
	\begin{scriptsize}
	\draw [fill=qqwuqq] (-6.34,4.35) circle (4.5pt);
	\draw (-6.04,5) node {\large $v_1$};
	\draw [fill=qqwuqq] (-7.6,2.67) circle (4.5pt);
	\draw (-7.6,3.32) node {\large $v_2$};
	\draw [fill=qqwuqq] (-4.44,0.47) circle (4.5pt);
	\draw (-4.04,1.02) node {\large $v_{4}$};
	\draw [fill=qqwuqq] (-6.54,0.51) circle (4.5pt);
	\draw (-6.24,1.06) node {\large $v_{3}$};
	\draw [fill=qqwuqq] (-4.44,4.29) circle (4.5pt);
	\draw (-4.14,4.94) node {\large $v_{6}$};
	\end{scriptsize}
	\end{tikzpicture}
}

    \caption{}
  \end{subfigure}
  \hspace{1cm}
  \begin{subfigure}[b]{0.26\linewidth}
  \centering
  \resizebox{\linewidth}{!}{
  
	\centering
	\definecolor{ccqqqq}{rgb}{0.8,0,0}
	\definecolor{qqwuqq}{rgb}{0,0.39215686274509803,0}
	\begin{tikzpicture}[scale=1,line cap=round,line join=round,>=triangle 45,x=1cm,y=1cm]
	\draw [line width=2pt] (-10.92,-2.56)-- (-9.66,-0.88);
	\draw [line width=2pt] (-10.92,-2.56)-- (-9.86,-4.72);
	\draw [line width=2pt] (-7.76,-2.59)-- (-7.76,-0.94);
	\draw [line width=2pt] (-7.76,-2.59)-- (-7.76,-4.76);
	\draw [line width=2pt] (-9.66,-0.88)-- (-7.76,-2.59);
	\draw [line width=2pt] (-10.92,-2.56)-- (-7.76,-2.59);
	\begin{scriptsize}
	\draw [fill=qqwuqq] (-9.66,-0.88) circle (4.5pt);
	\draw (-9.5,-0.28) node {\large $v_1$};
	\draw [fill=qqwuqq] (-10.92,-2.56) circle (4.5pt);
	\draw (-11.1,-2.05) node {\large $v_2$};
	\draw [fill=ccqqqq] (-7.76,-2.59) circle (4.5pt);
	\draw[color=ccqqqq] (-7.28,-2.18) node {\large $x_1$};
	\draw [fill=qqwuqq] (-7.76,-4.76) circle (4.5pt);
	\draw (-7.48,-4.2) node {\large $v_4$};
	\draw [fill=qqwuqq] (-9.86,-4.72) circle (4.5pt);
	\draw (-9.58,-4.16) node {\large $v_3$};
	\draw [fill=qqwuqq] (-7.76,-0.94) circle (4.5pt);
	\draw (-7.48,-0.28) node {\large $v_6$};
	\end{scriptsize}
	\end{tikzpicture}
}

    \caption{}
  \end{subfigure}

  \caption{The (a) original graph, (b) fill graph, and (c) component graph
  after pivoting $v_5$ and $v_7$.}
  \label{fig:graph_defs}
\end{figure}

For example, let us consider Figure~\ref{fig:graph_defs}.
The original graph $G$ has seven vertices $v_1, v_2,
\dots, v_7$, and the algorithm decides to pivot $v_5$ and $v_7$ marked in red. 
Eliminating these vertices induces a clique on
$v_1,v_2,v_4,v_6$ in the fill graph because
each pair of vertices is connected by a path through
the eliminated vertices $v_5$ and $v_7$.
Our algorithms implicitly maintain the fill graph
by maintaining the component graph, where $v_5$ and $v_7$ merge to form the
connected component~$x_1$ with edges incident to all remaining neighbors of $v_5$
and $v_7$ in the original graph.
Note that an upper bound for the number of edges in a component graph
is the number of edges in the original graph.
We repeatedly exploit this property when proving the time and space bounds
of our algorithms.

\subsection{Minimum Degree Orderings}

The minimum degree algorithm is a greedy heuristic for
reducing the cost of solving sparse linear systems that repeatedly eliminates
the variable involved in the fewest number of equations~\cite{GeorgeL89}.
Although there are many situations where this is suboptimal, it is
remarkably effective and widely used in practice.
For example, the approximate minimum degree algorithm (AMD)~\cite{AmestoyDD96}
is a heuristic for generating minimum degree orderings
that plays an integral role in the sparse linear algebra packages
in MATLAB~\cite{Matlab17}, Mathematica~\cite{Mathematica}, and Julia~\cite{Julia12}.

For any elimination ordering $(u_1, u_2, \dots, u_n$),
we let $G_i$ be the graph with vertices $u_1,u_2, \dots, u_{i}$
marked as eliminated
and $u_{i+1}, u_{i+2}, \dots, u_{n}$ marked as remaining.
We denote the corresponding sequence of fill graphs by
$(\Gfill{0}, \Gfill{1},\dots,\Gfill{n})$,
where $\Gfill{0} = G$ and $\Gfill{n}$ is the empty graph.
Throughout the paper, we frequently use the notation $[n] = \{1,2,\dots,n\}$
when iterating over sets.

\begin{definition}
A \emph{minimum degree ordering} is an elimination ordering such that
for all $i \in [n]$, the vertex
$u_i$ has minimum fill degree in $\Gfill{i - 1}$. Concretely,
this means that
\begin{align*}
\degfill{i-1}(u_{i}) = \min_{v \in \Vfill{i-1} } \degfill{i-1} ( v ).
\end{align*}
\end{definition}

The data structures we use for finding the vertices with minimum fill degree 
are randomized, so we need to be careful to not introduce dependencies
between different steps of the algorithm when several vertices are of minimum degree.
To avoid this problem, we simply require that the lexicographically-least vertex
be eliminated in the event of a tie.

Our notion for approximating a minimum degree ordering is based on finding
a vertex at each step whose degree is close to the minimum in $\Gfill{t}$.
Note that this is the goal of the AMD algorithm.

\begin{definition}
\label{def:ApproxMinDegree}
A \emph{$(1 + \epsilon)$-approximate greedy minimum degree ordering} is an elimination
ordering such that at each step $i \in [n]$, we have
\begin{align*}
  \degfill{i-1}(u_i)
  \leq
  \left( 1 + \epsilon \right)
  \min_{v \in \Vfill{i-1} }
  \degfill{i-1} ( v ).
\end{align*}
\end{definition}
\noindent
This decision process has no lookahead, and thus does not in any way
approximate the minimum possible total fill incurred during Gaussian
elimination, which is known to be NP-complete~\cite{Yannakakis81}.

\subsection{Related Works}
\label{subsec:Related}

\paragraph*{Gaussian Elimination and Fill.}
The study of pivoting orderings is a fundamental question in combinatorial
scientific computing.
Work by George~\cite{George73} led to the study of
nested dissection algorithms, which utilize separators
to give provably smaller fill bounds for planar~\cite{RoseTL76,LiptonRT79}
and separable graphs~\cite{GilbertT87,AlonY10}.
A side effect of this work is the better
characterization of fill via component graphs~\cite{Liu85},
which is used to compute the total fill-in of a specific elimination
ordering~\cite{GilbertNP94}.
This characterization is also used to construct elimination trees,
which are ubiquitous in scientific computing
to preallocate memory and optimize cache behaviors~\cite{Liu90}.

\paragraph*{Finding Low-Fill Orderings.}
The goal of an elimination ordering is to minimize the total fill.
Unfortunately, this problem is NP-complete~\cite{Yannakakis81,Berman90}.
Algorithms that approximate the minimum fill-in within polynomial factors
have been studied~\cite{NatanzonSS00}, as well as
algorithms~\cite{KaplanST99,FominV13} and hardness
results~\cite{WuAPL14,Bliznets16,CaoS17} for parameterized variants.
Partially due to the high overhead of the previous algorithms,
the minimum degree heuristic remains as one of the most widely-used
methods for generating low-fill orderings~\cite{GeorgeL89}.

Somewhat surprisingly, we were not able to find prior works that
compute the minimum degree ordering in time faster
than $O(n^3)$ or works that utilize the implicit representation
of fill provided by elimination trees.\footnote{
We use speculative language here due to the vastness
of the literature on variants of minimum degree algorithms.
}
On the other hand, there are various heuristics for finding
minimum degree-like orderings, including
multiple minimum degree (MMD)~\cite{Liu85}
and the approximate minimum degree algorithm (AMD)~\cite{AmestoyDD96}.
While both of these methods run extremely well in practice,
they have theoretically tight performances of $\Theta(n^2m)$ for MMD
and $\Theta(nm)$ for AMD~\cite{HeggernesEKP01}.
Furthermore, AMD is not always guaranteed to produce a vertex of approximate minimum degree.

\paragraph*{Removing Dependencies in Randomized Algorithms.}
Our size estimators are dynamic---the choice of the pivot, which
directly affects the subsequent fill graph, is a result of the randomness
used to generate the pivot in the previous step---and prone to
dependency problems.
Independence between the access sequence and internal randomness is
a common requirement in recent works on data structures
for maintaining spanning trees, spanners,
and matchings~\cite{BaswanaGS15,KapronKM13,Solomon16}.
Often these algorithms only have guarantees
in the oblivious adversary model,
which states that the adversary can choose the graph and the sequence of
updates, but it cannot choose updates adaptively in response to the
randomly-guided choices of the algorithm.

Recent works in randomized dimensionality-reduction have approached this
issue of dependency by injecting additional randomness to preserve
independence~\cite{LeeS15}.
Quantifying the amount of randomness that is ``lost'' over the course of an
algorithm has recently been characterized using mutual
information~\cite{KapralovNPWWY17}, but their results do not allow
for us to consider $n$ adversarial vertex pivots.
Our analysis also has tenuous connections to recent works utilizing 
matrix martingales to analyze repeated introductions of randomness
into graph sparsification algorithms~\cite{KyngS16,KyngPPS17}.

%% file: Overview.tex
\section{Overview}
\label{sec:Overview}

We discuss the main components of our algorithms in three parts.
In Section~\ref{subsec:OverviewSketching}
we explore how dynamic graph sketching can be used to construct a
randomized data structure that maintains approximate degrees under
vertex eliminations.
In Section~\ref{subsec:Correlation} we demonstrate how data structures
that work against oblivious adversaries can fail against adaptive adversaries.
We also describe our approach to circumvent this problem for approximate minimum
degree sequences.
In Section~\ref{subsec:OverviewDegreeEstimation} we discuss a local degree estimation
routine (the new key primitive) in the context of estimating
the number of nonzero columns of a matrix via sampling.
Finally, in Section~\ref{subsec:OverviewComputing} we explain the implications
of our results to the study of algorithms for computing elimination orderings.

\input{Overview_Sketching}
\input{Overview_Decorrelation}
\input{Overview_DegreeEstimation}
\input{Overview_Computing}

%% file: Overview_Sketching.tex
\subsection{Dynamically Sketching Fill Graphs}
\label{subsec:OverviewSketching}

The core problem of estimating fill degrees
can be viewed as estimating the cardinality of sets
undergoing unions and deletion of elements.
Cardinality estimation algorithms in the streaming algorithm literature
often trade accuracy for space~\cite{FlajoletM85,CormodeM05},
but our degree-approximation data structures use sketching to trade
space for accuracy and more efficient update operations.

We first explain the connection between computing fill degrees and
estimating the size of reachable sets.
Assume for simplicity that no edges exist between the remaining vertices
in the component graph $\Gcomp{}$.
Split each remaining vertex $u$ into two vertices $u_1$ and $u_2$,
and replace every edge $(u,x)$ to a component vertex $x$ by the directed
edges $(u_1,x)$ and $(x,u_2)$.
The fill degree of~$u$ is the number of remaining
vertices $v_2$ reachable from $u_1$ (not including $u_1$).
Cohen~\cite{Cohen97} developed a nearly-linear time size-estimation framework
for reachability problems using sketching and $\ell_0$-estimators.
Adapting this framework to our setting for fill graphs
leads to the following kind of $\ell_0$-sketch data structure.
We refer to the set $N(u)\cup\{u\}$ as the \emph{$1$-neighborhood}
of $u$, and we call its cardinality $\deg(u)+1$ the \emph{$1$-degree} of $u$. 
\begin{definition}\label{def:Sketch}
A \emph{1-neighborhood $\ell_0$-sketch} of a graph $G$ is constructed as follows:
\begin{enumerate}
  \item Each vertex $u \in V$ independently generates a random key $R(u)$
    uniformly from $[0,1)$.
  \item Then each vertex determines which of its neighbors
    (including itself) has the smallest key.
    We denote this by the function
    \[
      \minimizer(u) \defeq \argmin_{v \in N(u)\cup\{u\}} R(v).
    \]
\end{enumerate}
\end{definition}

To give some intuition for how sketching is used to estimate
cardinality, observe that choosing keys independently and uniformly at
random essentially assigns a random vertex $N(u)\cup\{u\}$ to be $\minimizer(u)$.
Therefore, the key value $R(\minimizer(u))$ is correlated with $\deg(u)+1$.
This correlation is the cornerstone of sketching.
If we construct $k=\Omega(\log{n}\epsilon^{-2})$ independent sketches,
then by concentration we can use an order statistic of $R_i(\minimizer(u))$
over all $k$ sketches to give an $\epsilon$-approximation of
$\deg(u) + 1$ with high probability.
We gives the full details in Appendix~\ref{sec:SketchingProofs}.

To maintain sketches of the fill graph as it undergoes vertex
eliminations, we first need to implicitly maintain the component graph
$\Gcomp{}$ (Lemma~\ref{lem:DegreeEstimationDS}).
We demonstrate how to efficiently propagate key values
in a sketch as vertices are pivoted in Section~\ref{sec:DynamicGraphs}.
For now, it is sufficient to know that each vertex in a sketch has an
associated min-heap that it uses to report and update its minimizer.
Because eliminating vertices leads to edge contractions in the component
graph, there is an additional layer of intricacies that we need to resolve
using amortized analysis.

Suppose $v$ is the vertex eliminated as we go from $\Gcomp{t}$ to $\Gcomp{t+1}$.
The sketch propagates this information to relevant vertices in the graph
using a two-level notification mechanism.
The neighbors of~$v$ are informed first, and then they notify their neighbors
about the change, all the while updating the key values in their heaps.
We outline the subroutine $\textsc{PivotVertex}(v)$ that accomplishes this:
\begin{enumerate}
  \item Update the min-heaps of every remaining neighbor of $v$.
  \item For each component neighbor $x$ of $v$,
        if the minimum in its heap changes, then
        propagate this change to the remaining neighbors of $x$
        and merge $x$ with $v$.
\end{enumerate}
While it is simple enough to see that this algorithm correctly maintains key
values,
bounding its running time is nontrivial and requires a careful amortized
analysis to show that the bottleneck operation is the merging of component vertices.

We can merge two min-heaps in $O(\log^2 n)$ time, so 
merging at most $n$ heaps takes $O(n \log^2 n)$ time in total.
To bound the cost of heap updates due to merges,
we define the potential of the component graph as
\[
  \Phi(\Gcomp{t}) \defeq \sum_{u \in \Vcomp{,t}} D(u) \log (D(u)),
\]
where $D(u)$ is the sum of the original degrees of vertices merged into $u$.
Using the fact that a merge operation only informs neighbors
of at most one of the two merged vertices,
we are able to show that the number of updates to produce $\Gcomp{t}$
is of the order of $\Phi(\Gcomp{t})$. It follows that the total number of
updates is at most $O(m \log n)$, which gives us a total update cost of $O(m
\log^2 n)$.

%% file: Overview_Decorrelation.tex
\subsection{Correlation and Decorrelation}
\label{subsec:Correlation}

We now discuss how we use the randomized sketching data structure within a
greedy algorithm.
We start with a simple concrete example to illustrate a problem that an
adaptive adversary can cause.
Consider a data structure that uses sketching to estimate the cardinality of a
subset $S \subseteq [n]$ under the insertion and deletion of elements.
This data structure randomly generates a subset of keys $T
\subseteq [n]$ such that $|T| = \Theta(\log{n} \epsilon^{-2})$, and
it returns as its estimate the scaled intersection
\[
  n \cdot \frac{\left|S \cap T\right|}{|T|},
\]
which is guaranteed to be within an $\epsilon n$-additive error of the
true value $|S|$ by Chernoff bounds, assuming that $T$ is generated
independently of $S$.
Clearly this cardinality-estimation algorithm works in the oblivious adversary model.

However, an adaptive adversary can use answers to previous queries
to infer the set of secret keys~$T$ in $O(n)$ updates and queries.
Consider the following scheme 
in Figure~\ref{fig:ProblematicUpdates}
that returns $S=T$.

\begin{figure}[H]
\begin{algbox}
\begin{enumerate}
\item Initialize $S=[n]$.
\item For each $i=1$ to $n$:
  \begin{enumerate}
    \item Delete $i$ from $S$. If the estimated size of $S$ changed, reinsert $i$ into $S$.
  \end{enumerate}
\item Return $S$.
\end{enumerate}
\end{algbox}
\caption{An adaptive routine that amplifies the error of a
  cardinality-estimation scheme that uses a fixed underlying sketch.}
\label{fig:ProblematicUpdates}
\end{figure}

While the updates performed by a greedy algorithm are less extreme than
this, in the setting where we maintain the cardinality of the smallest of
$k$ dynamic sets, having access to elements in the minimizer does allow for
this kind of sketch deduction.
Any accounting of correlation (in the standard sense) also allows for
worst-case kinds of adaptive behavior, similar to the scheme above.

To remove potential correlation, we use an external routine that is analogous to
the local degree-estimation algorithm used in the approximate
minimum degree algorithm, which runs in time close to the degree it estimates.
In this simplified example, suppose for each cardinality query that the data structure
first regenerates $T$.
Then the probability that $i$ belongs to $S$ is $\Theta(\log{n}\epsilon^2 / n)$.
Stepping through all $i \in [n]$, 
it follows that the expected number of deletions is $\Theta(\log{n}
\epsilon^{-2})$, and hence $S$ remains close to size $n$ with high probability.

Reinjecting randomness is a standard method for
decorrelating a data structure across steps.
However, if we extend this example to the setting where we maintain the
cardinality of $k$ sets (similar to our minimum degree algorithm),
then the previous idea requires that we reestimate the size of every set
to determine the one with minimum cardinality.
As a result, this approach is prohibitively expensive.
However, these kinds of cardinality estimations are actually local---meaning that 
it is sufficient to instead work with a small and accurate subset of candidates
sets.
If we compute the set with minimum cardinality among the candidates
using an external estimation scheme,
then this decision is independent of the random choice of $T$ in the sketching
data structure, which then allows us to use the sketching data structure to
generate the list candidates.

Our algorithm for generating an approximate greedy minimum degree ordering
relies on a similar external routine called
$\textsc{EstimateFill1Degree}(u,\epsilon)$,
which locally estimates the fill 1-degree of~$u$ at any step of the algorithm
in time proportional to $\deg(u)$ in the original graph.
We further describe this estimator in
Section~\ref{subsec:OverviewDegreeEstimation} and present the full
sampling algorithm in Section~\ref{sec:DegreeEstimation}.
In Section~\ref{sec:Decorrelation} we show that to generate an approximate
greedy minimum degree sequence, it is instead sufficient to pivot the vertex
\[
  \argmin_{u \in \Vfill{}}
  \left(1 - \frac{\epsilon \cdot \Exp(1)}{\log{n}}\right)
  \cdot \textsc{EstimateFill1Degree}\left(u, \frac{\epsilon}{\log{n}}\right)
\]
at each step, which we call the \emph{$\epsilon$-decayed minimum} over all
external estimates.

Analogous to the discussion about the set cardinality estimation above,
evaluating the degrees of every remaining vertex using
$\textsc{EstimateFill1Degree}$ at each step
is expensive and leads to a total cost of $\Omega(nm)$.
However, we can reincorporate the sketching data structure and use the following
observations about the perturbation coefficient involving the
exponential random variable $\Exp(1)$ to sample a small number of candidate
vertices that contains the $\epsilon$-decayed minimum.
\begin{itemize}
\item For a set of vertices whose degrees are
  within $1 \pm \epsilon/\log{n}$ of each other, 
it suffices to randomly select and consider $O(1)$ of them 
by generating the highest order statistics of exponential random variables in
decreasing order.
\item By the memoryless property of the exponential distribution,
if we call $\textsc{EstimateFill1Degree}$, then with constant probability it
will be for the vertex we pivot.
Therefore, we can charge the cost of these evaluations to the original
edge count and retain a nearly-linear running time.
\end{itemize}

\noindent
Invoking $\textsc{EstimateFill1Degree}$ only on the candidate vertices
allows us to efficiently find the $\epsilon$-decayed minimizer in each step, which
leads to the nearly-linear runtime as stated in
Theorem~\ref{thm:main}.
The key idea is that any dependence on the $\ell_0$-sketches stops after the
candidates are generated, since their degrees only depend on the randomness of
an external cardinality-estimation routine.

%% file: Overview_DegreeEstimation.tex
\subsection{Local Estimation of Fill Degrees}
\label{subsec:OverviewDegreeEstimation}

A critical part of the approximate min-degree algorithm
is the \textsc{EstimateFill1Degree} function,
which estimates the fill 1-degree of a vertex $u \in \Vfill{}$
using fresh randomness and $O(\deg(u) \log^2{n} \epsilon^{-2})$
oracle queries to the component graph $\Gcomp{}$.
At the beginning of Section~\ref{sec:DegreeEstimation} we show how to construct
a $(0,1)$-matrix $A$ where each row corresponds to a remaining neighborhood
of a component neighbor of $u$. The number of nonzero columns in $A$
is equal to $\degfill{}(u)$.
Using only the following matrix operations (which correspond to component
graph oracle queries), we analyze the more general problem of counting
the number of nonzero columns in a matrix.
We note that this technique is closely related to recent results in wedge sampling
for triangle counting~\cite{KallaugherP17,EdenLRS17}.

\begin{itemize}
  \item $\textsc{RowSize}(A,i)$: Returns the number of nonzero elements in row $i$ of $A$.
  \item $\textsc{SampleFromRow}(A,i)$: Returns a column index $j$ uniformly at
    random from the nonzero entries of row $i$ of $A$.
  \item $\textsc{QueryValue}(A,i,j)$: Returns the value of $A(i,j)$.
\end{itemize}

\begin{restatable*}[]{lemma}{NonZeroColumnEstimator}
\label{lem:NonZeroColumnEstimator}
There is a routine $\textsc{EstimateNonzeroColumns}$ using
the three operations above that takes as input
(implicit) access to a matrix $A$
and an error $\epsilon$,
and returns an $\epsilon$-approximation to the number
of nonzero columns in $A$ with high probability.
The expected total number of operations used is
$O(r \log^2{n} \epsilon^{-2})$,
where $r$ is the number of rows and $n$ is the number of columns in $A$.
\end{restatable*}

We now give an overview of how \textsc{EstimateNonzeroColumns} works.
Let $B$ be the normalized version of $A$ where every nonzero entry is
divided by its column sum.
The sum of the nonzero entries in $B$ is the number of nonzero
columns in $A$, denoted by $\textsc{NonzeroColumns}(A)$.
If we uniformly sample a nonzero entry of $B$, then the mean of
this distribution is $\textsc{NonzeroColumns}(A)/\nnz(A)$.
Because random variables sampled from this distribution take their value in $[0,1]$,
we can estimate their mean using an \textsc{EstimateMean} subroutine (Lemma~\ref{lem:MeanEstimation}),
which does the following:
\begin{enumerate}
	\item Set a threshold $\sigma$ depending on the accuracy of the desired estimate.
  \item Sample $k$ independent random variables from the distribution until their sum first exceeds $\sigma$.
	\item Return $\sigma/k$.
\end{enumerate}

Using the matrix operations above,
we can easily sample indices $(i,j)$ of nonzero entries in $B$,
but evaluating $B(i,j)$ requires that we know the $j$-th column sum of $A$.
Therefore, to compute this column sum we estimate the mean of a
Bernoulli distribution on the $j$-th column of~$A$ defined by selecting an
entry from $A[:,j]$ uniformly at random. This distribution has
mean $\nnz(A[:,j])/r$, and it is amenable to sampling using the provided
operations.

While the previous estimator works satisfactorily, we show how to combine
these distributions and use a hitting time argument to reduce
the sample complexity by a factor of $O(\epsilon^{-2})$.
Specifically, for a fixed column, we consider a random variable that has a
limited number of attempts to find a nonzero entry by 
uniformly sampling rows.
By optimizing the number of attempts, we can reduce our error overhead in the
runtime
at the expense of a $1/\text{poly}(n)$ perturbation to the approximation.

%% file: Overview_Computing.tex
\subsection{Significance to Combinatorial Scientific Computing}
\label{subsec:OverviewComputing}

Despite the unlikelihood of theoretical gains for solving linear systems
by improved direct methods for sparse Gaussian elimination,
we believe our study could influence combinatorial scientific computing in
several ways.
First, we provide evidence in Section~\ref{sec:Hardness} for the nonexistence
of nearly-linear time algorithms for finding exact minimum degree orderings by
proving conditional hardness results.
Our reduction uses the observation that determining if a graph can be covered
by a particular union of cliques (or equivalently, that the fill graph is
a clique after eliminating certain vertices) is equivalent to the
orthogonal vectors problem~\cite{Williams05}.
Assuming the strong exponential time hypothesis,
this leads to a conditional hardness of $\Omega(m^{4/3-\theta})$
for computing a minimum degree ordering.
However, we believe that this result is suboptimal and 
that a more careful construction could lead to $\Omega(n m^{1-\theta})$-hardness.

On the other hand, advances in minimum degree algorithms cannot be justified in
practice solely by worst-case asymptotic arguments.
In general, nested dissection orderings are asymptotically superior in
quality to minimum degree orderings~\cite{hendrickson1998improving}.
Furthermore, methods based on Krylov spaces, multiscale
analysis, and iterative
methods~\cite{gutknecht2007brief,gaul2013framework} are becoming increasingly popular
as they continue to improve state-of-the-art solvers for large sparse systems.
Such advancements are also starting to be reflected in theoretical works.
As a result, from both a theoretical and practical perspective, we believe that
the most interesting question related to minimum degree algorithms is whether
or not such sequences lead to computational gains for problems of moderate size.

In our approximate minimum degree algorithm, the $O(\log^5{n})$ term
and convenient choice of constants
preclude it from readily impacting elimination orderings in practice.
However, the underlying sketching technique is quite flexible.
For example, consider modifying the dynamic $\ell_0$-sketches such that:
\begin{enumerate}
\item Each vertex maintains $k$ random numbers in the range $[0, 1)$.
\item Each component vertex maintains the smallest $k$ numbers of its remaining
  neighbors.
\item Each remaining vertex then maintains the smallest $k$ numbers among its
  component neighbors.
\end{enumerate}
If we repeatedly eliminate the vertex whose median is largest, this routine
is similar to using~$k$ copies of the previous type of the sketch.
Letting $k=\Omega(\log{n} \epsilon^{-2})$,
we can analyze this variant against an oblivious adversary 
using slight modifications to our original sketching algorithm~\cite{massart2000constants}.
Although our analysis demonstrates that new tools are necessary for studying
its behavior within a greedy algorithm,
we experimentally observed desirable behavior for such sketches.
Therefore, we plan to continue studying this kind of adaptive graph sketching
both theoretically and experimentally.

%% file: Sketching.tex
\section{Sketching Algorithms for Computing Degrees}
\label{sec:Sketching}

Let us recall a few relevant definitions from Section~\ref{sec:Preliminaries}
for convenience.  For a given vertex elimination sequence
$(u_1, u_2, \dots, u_n)$,
let $\Gfill{t}$ denote the fill graph obtained by pivoting
vertices $u_1, u_2, \dots, u_t$, and
let $\delta_t$ denote the minimum degree in $\Gfill{t}$.
An $\ell_0$-sketch data structure consists of the following:
\begin{enumerate}
\item Each vertex $u$ independently generates a key $R(u)$ from $[0,1)$ uniformly at random.
\item Then each vertex $u$ determines which neighbor (including itself) has the smallest
key value. We denote this neighbor by $\minimizer(u)$.
\end{enumerate}

In this section we show that if an $\ell_0$-sketch can efficiently be
maintained for a dynamic graph, then we can use the same set of sketches at
each step to determine the vertex with minimum fill degree and eliminate it.
We explore the dynamic $\ell_0$-sketch data structure for efficiently
propagating key values under pivots in detail in
Section~\ref{sec:DynamicGraphs} (and for now we interface it via
Theorem~\ref{thm:DataStructureMain}).
This technique leads to improved algorithms for computing the minimum degree
ordering of a graph, which we analyze in three different settings.

First, we consider the case where the minimum degree at each step is bounded.
In this case
we choose a fixed number of $\ell_0$-sketches and keep track of every 
minimizer of a vertex over all of the sketch copies.
Note that we can always use $n$ as an upper bound on the minimum fill degree.

\begin{theorem}
\label{thm:BoundedDegreeAlgo}
There is an algorithm \textsc{DeltaCappedMinDegree} that,
when given a graph with a lexicographically-first
min-degree ordering whose minimum degree is always bounded by $\Delta$,
outputs this ordering with high probability in expected time $O(m \Delta \log^3 n)$
and uses space $O(m \Delta \log n)$.
\end{theorem}
	
Next, we relax the bound on the minimum degrees over all steps of the algorithm
and allow the time and space complexity to be output sensitive by
adaptively increasing the number of $\ell_0$-sketches
as the algorithm progresses.

\begin{theorem}
\label{thm:OutputSensitiveAlgo}
There is an algorithm \textsc{OutputSensitiveMinDegree} that,
when given a graph with a lexicographically-first
min-degree sequence $(\delta_1, \delta_2, \dots, \delta_n)$,
outputs this ordering with high probability in expected time
$O(m \cdot \max_{t \in [n]} \delta_t \cdot \log^3 n)$
and uses space $O(m \cdot \max_{t \in [n]} \delta_t \cdot \log n)$.
\end{theorem}

Lastly, we modify the algorithm to compute an
approximate minimum degree vertex at each step.
By maintaining $\Theta(\log n \epsilon^{-2})$ copies
of the $\ell_0$-sketch data structure, we are able to
accurately approximate the 1-degree of a vertex using the
$(1-1/e)$-th\footnote{Note that we use $e$ to refer to the base of the natural logarithm.}
order statistic of the key values of its minimizers.
We abstract this idea using the following approximate degree data structure,
which when given an elimination ordering directly leads to a nearly-linear time
algorithm.

\begin{theorem}
\label{thm:ApproxDegreeDS}
There is a data structure $\variable{ApproxDegreeDS}$
that supports the following methods:
\begin{itemize}
  \item $\textsc{ApproxDegreeDS\_Pivot}(u)$, which pivots a remaining vertex $u$.
  \item $\textsc{ApproxDegreeDS\_Report}()$, which provides
  balanced binary search tree (BST) containers
  $V_1, V_2, \dots, V_{B}$
  such that all the vertices in the bucket $V_{i}$ have 1-degree in the range
  \[
    \left[
    \left( 1 + \epsilon \right)^{i - 2},
    \left( 1 + \epsilon \right)^{i + 2}
    \right].
  \]
\end{itemize}
  
\noindent
The memory usage of this data structure is $O(m \log{n} \epsilon^{-2})$.
Furthermore, if the pivots are picked independently from the
randomness used in this data structure
(i.e., we work under the oblivious adversary model) then:
\begin{itemize}
  \item The total cost of all the calls to
     $\textsc{ApproxDegreeDS\_Pivot}$ is bounded by $O(m \log^3{n} \epsilon^{-2})$.
  \item
    The cost of each call to $\textsc{ApproxDegreeDS\_Report}$
  is bounded by $O(\log^2{n} \epsilon^{-1})$.
\end{itemize}
\end{theorem}

\subsection{Computing the Exact Minimum Degree Ordering}
We first consider the case where the minimum degree in each of the
fill graphs $\Gfill{t}$ is at most $\Delta$.  In this case,
we maintain $k = O(\Delta \log n)$ copies of the $\ell_0$-sketch data structure.
By a coupon collector argument, any vertex with degree at most $\Delta$ contains
all of its neighbors in its list of minimizers with high probability.
This implies that for each $t \in [n]$, we can obtain the exact
minimum degree in~$\Gfill{t}$ with high probability.
Figure~\ref{fig:DeltaCappedGlobalVar} briefly describes the data structures we
will maintain for this version of the algorithm.

\begin{figure}[H]
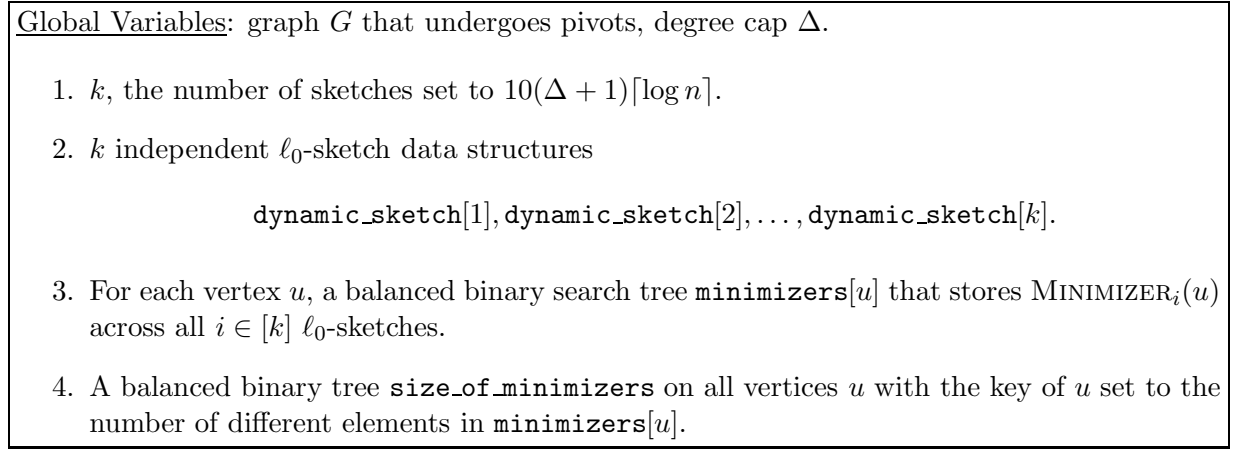

\begin{algbox}
  \underline{Global Variables}:
  graph $G$ that undergoes pivots,
  degree cap $\Delta$.

  \begin{enumerate}
    \item $k$, the number of sketches set to $10(\Delta+1)\lceil\log n\rceil$. 
    \item $k$ independent $\ell_0$-sketch data structures
    \[
      \variable{dynamic\_sketch}[1],
      \variable{dynamic\_sketch}[2],
      \dots,
      \variable{dynamic\_sketch}[k].
    \]

    \item For each vertex $u$, a balanced binary search
      tree $\variable{minimizers}[u]$ that stores $\minimizer_i(u)$
      across all $i \in [k]$ $\ell_0$-sketches.

    \item A balanced binary tree $\variable{size\_of\_minimizers}$
      on all vertices $u$ with the key of $u$ set to the
      number of different elements in $\variable{minimizers}[u]$.
  \end{enumerate}
\end{algbox}
\caption{Global variables for the $\Delta$-capped min-degree algorithm
\textsc{DeltaCappedMinDegree}.}
\label{fig:DeltaCappedGlobalVar}
\end{figure}

Note that if we can efficiently maintain the data structures in
Figure~\ref{fig:DeltaCappedGlobalVar}, then querying the minimum element
in $\variable{size\_of\_minimizers}$ returns the (lexicographically-least) vertex
with minimum degree. Theorem~\ref{thm:DataStructureMain} demonstrates that we
can maintain the $\ell_0$-sketch data structures efficiently.

\begin{theorem}
\label{thm:DataStructureMain}
  Given i.i.d.\ random variables $R(v)$ associated with each vertex
  $v \in \Vfill{t}$,
  there is a data structure $\variable{DynamicSketch}$ that, for each vertex
  $u$, maintains the vertex with minimum~$R(v)$
  among itself and its neighbors in $\Gfill{t}$.
  This data structure supports the following methods:
  \begin{itemize}
    \item $\textsc{QueryMin}(u)$, which returns
      $\minimizer(u)$
      for a remaining vertex $u$ in $O(1)$ time.
    \item $\textsc{PivotVertex}(u)$, which pivots a remaining vertex $u$
    and returns the list of all remaining vertices $v$ whose value
    of
    $\minimizer(v)$
    changed immediately after this pivot.
  \end{itemize}
  The memory usage of this data structure is $O(m)$.
  Moreover, for any choice of key values $R(v)$:
  \begin{itemize}
    \item The total cost of all the pivots is $O(m \log^2{n})$.
    \item 
      The total size of all lists returned by
      $\textsc{PivotVertex}$ over all steps is $O(m \log{n})$.
  \end{itemize}
\end{theorem}

\noindent
This theorem relies on intermediate data structures described in
Section~\ref{sec:DynamicGraphs}, so we defer the proof until the end of that
section.
Note that this \variable{DynamicSketch} data structure will be essential
to all three min-degree algorithms.

Now consider a sketch of $G^+$ and a vertex $u$ with degree
$\degfill{}(u) \leq \Delta$.  By symmetry of the~$R(v)$ values, each vertex in
$\Nfill{}(u)\cup\{u\}$ is the minimizer of $u$ with probability $1 / (\degfill{}(u)+1)$.
Therefore, if we maintain $O(\Delta \log{n})$ independent $\ell_0$-sketches, we
can ensure that we have an accurate estimation of the minimum fill degree
with high probability.
The pseudocode for this routine is given in Figure~\ref{fig:DeltaCapped}.
We formalize the probability guarantees in Lemma~\ref{lem:approxdegree1} and
Lemma~\ref{lem:approxdegree2},
which are essentially a restatement of \cite[Theorem 2.1]{Cohen97}.

\begin{figure}[H]
\begin{algbox}
$\textsc{DeltaCappedMinDegree}(G, \Delta)$

\underline{Input}: graph $G=(V,E)$, threshold $\Delta$.

\underline{Output}: exact lexicographically-first min-degree ordering
$(u_1, u_2, \dots, u_n)$.

\begin{enumerate}
\item For each step $t=1$ to $n$:
  \begin{enumerate}
    \item Set $u_t \leftarrow \min(\variable{size\_of\_minimizers})$.
    \item $\textsc{DeltaCappedMinDegree\_Pivot}(u_t)$.
    \end{enumerate}
  \item Return $(u_1,u_2,\dots,u_n)$.
\end{enumerate}
\end{algbox}

\begin{algbox}
$\textsc{DeltaCappedMinDegree\_Pivot}(u)$

\underline{Input}:
vertex to be pivoted $u$.

\underline{Output}:
updated global state.

\begin{enumerate}

\item For each sketch $i=1$ to $k$:
\begin{enumerate}
\item $\left(v_1, v_2, \ldots, v_{\ell} \right)
  \leftarrow \variable{dynamic\_sketch}[i].\textsc{PivotVertex}(u)$,
  the set of vertices in the $i$-th sketch whose minimizers changed after pivoting
    out $u$.
\item For each $j=1$ to $\ell$:
\begin{enumerate}
\item Update the values corresponding to sketch $i$
  in $\variable{minimizers}[v_j]$.
\item Update the entry for $v_j$ in
$\variable{size\_of\_minimizers}$ with the
size of $\variable{minimizers}[v_j]$.
\end{enumerate}
\end{enumerate}
\end{enumerate}

\end{algbox}

\caption{Pseudocode for the exact $\Delta$-capped min-degree algorithm,
which utilizes the global data structures for \textsc{DeltaCappedMinDegree}
defined in Figure~\ref{fig:DeltaCappedGlobalVar}.}
\label{fig:DeltaCapped}
\end{figure}

\begin{lemma}\label{lem:approxdegree1}
With high probability,
for all remaining vertices $u$ such that
$\degfill{}(u) \le 2\Delta$ we have
\[
  \variable{size\_of\_minimizers}[u]=\degfill{}(u)+1.
\]
\end{lemma}

\begin{proof}
The only way we can have
$\variable{size\_of\_minimizers}[u] < \degfill{}(u)+1$ is if at least
one neighbor of $u$ or~$u$ itself is not present in
$\variable{minimizers}[u]$.  Let $v$ be an arbitrary vertex in
$\Nfill{}(u) \cup\{u\}$.
The probability of $v$ not being the minimizer in any of the
$k=10(\Delta+1)\lceil\log n \rceil$ sketches is
\begin{align*}
  \prob{}{\minimizer_i(u) \ne v \text{ for all $i \in [k]$}} 
    &= \left(1-\dfrac{1}{\degfill{}(u)+1} \right)^{k}\\
    &\le \left(1-\dfrac{1}{2\Delta+1}\right)^{10(\Delta+1)\log{n}} \\
    &\le \exp\left(-\frac{10(\Delta+1)\log n}{2\Delta + 1}\right) \\
    &\le \frac{1}{n^{5}}.
\end{align*}

\noindent
We can upper bound the probability that there exists a
vertex $v \in \Nfill{}(u) \cup \{u\}$ not in $\variable{minimizers}[u]$
using a union bound. It follows that
\begin{align*}
  \prob{}{\variable{size\_of\_minimizers}[u] < \degfill{}(u)+1} &\le 
    |\Nfill{}(u)\cup\{u\}| \cdot \prob{}{\minimizer_i(u) \ne v \text{ for all $i \in [k]$}}\\
    &\le \frac{1}{n^{4}}.
\end{align*}
Using a second union bound for the event that there exists
a vertex $u \in \Vfill{}$ such that $\degfill{}(u) \le 2\Delta$ and
$\variable{size\_of\_minimizers}[u] < \degfill{}(u)+1$ completes the proof.
\end{proof}

\begin{lemma}\label{lem:approxdegree2}
With high probability,
for all remaining vertices $u$ with $\degfill{}(u) > 2\Delta$
we have
\[
\variable{size\_of\_minimizers}[u] > \Delta+1.
\]
\end{lemma}

\begin{proof}
We first upper bound the probability of the event
$\variable{size\_of\_minimizers}[u] \le \Delta + 1$.
Let $S$ be any subset of $\Nfill{}(u)\cup\{u\}$ of size
$\degfill{}(u) - \Delta > 0$.
Using the assumption that $\degfill{}(u) \ge 2\Delta + 1$,
\begin{align*}
  \Pr\left[S \cap \variable{minimizers}[u] = \emptyset \right] 
    &= \left(1 - \frac{\degfill{}(u) - \Delta}{\degfill{}(u)+1}\right)^k\\
    &\le \exp\left(-\frac{10\left(\degfill{}(u) - \Delta\right)(\Delta+1)\log{n}}{\degfill{}(u)+1}\right)\\
    &\le \exp\left(-\frac{5\left(\degfill{}(u) + 1\right)(\Delta+1)\log{n}}{\degfill{}(u)+1}\right)\\
    &= \frac{1}{n^{5(\Delta+1)}}.
\end{align*}

\noindent
Next, sum over all choices of the set $S$ and use a union bound.
Note that this over counts events, but this suffices for an upper bound.
It follows that
\begin{align*}
  \Pr\left[\variable{size\_of\_minimizers}[u] \leq \Delta+1\right] & \leq 
  {\degfill{}(u)+1 \choose \degfill{}(u)-\Delta} \cdot \Pr\left[S \cap \variable{minimizers}[u] = \emptyset\right] \\
& \leq {\degfill{}(u)+1 \choose \Delta+1} \cdot \dfrac{1}{n^{5(\Delta+1)}} \\
& \leq \dfrac{n^{\Delta+1}}{n^{5(\Delta+1)}} \\
&\le \dfrac{1}{n^{4}}.
\end{align*}

\noindent
Using a second union bound for the event that there exists a vertex
$u \in \Vfill{}$ such that $\degfill{}(u) > 2\Delta$ and
$\variable{size\_of\_minimizers}[u] \le \Delta + 1$ completes the proof.
\end{proof}

\begin{proof}[Proof of Theorem~\ref{thm:BoundedDegreeAlgo}.]
The algorithm correctly pivots the minimum degree vertex by
Lemma~\ref{lem:approxdegree1} and Lemma~\ref{lem:approxdegree2}.
For the space complexity, each of the $k$ $\ell_0$-sketch data
structures uses $O(m)$ memory
by Theorem~\ref{thm:DataStructureMain}, and
for each vertex there is a corresponding balanced binary search tree
$\variable{minimizers}$ which uses
$O(k)$ space.
We also have $\variable{size\_of\_minimizers}$, which uses
$O(n)$ space.
Therefore, since we assume $m \ge n$, the total space is
$O(km+nk + n) = O(m \Delta \log{n})$.

For the running time, Theorem~\ref{thm:DataStructureMain} gives a cost
of $O(m \log^2{n})$ across all \textsc{PivotVertex} calls per sketch,
and thus a total cost of $O(m \Delta \log^3{n})$ over all sketches.
Theorem~\ref{thm:DataStructureMain} also states
the sum of~$\ell$ (the length
of the update lists) across all steps is at most $O(m \log{n})$.
Each of these updates leads to two BST updates,
so the total overhead is $O(m \log^2{n})$,
which is an equal order term.
\end{proof}

\subsection{Modifying the Algorithm to be Output Sensitive by Adaptive Sketching}
\label{subsec:AdaptiveSketching}

If we do away with the condition that minimum fill degrees are bounded above by
$\Delta$, then the number of copies of the $\ell_0$-sketch data structure
needed depends only on the values of the minimum fill degree at each
step.
Therefore, we can modify \textsc{DeltaCappedMinDegree} to potentially
be more efficient by adaptively maintaining the required number of sketches.

To accurately estimate degrees in $\Gfill{t}$ we need 
$\Omega(\delta_t \log n)$ copies of the $\ell_0$-sketch data structure,
but we do not know the values of $\delta_t$ a priori.
To rectify this, consider the following scheme that adaptively keeps a sufficient
number of copies of the $\ell_0$-sketch data structures.
First, initialize the value
$c=\delta_0$ (the minimum degree in $G$). Then for each step $t=1$ to $n$ update $c$
according to:
\begin{enumerate}
  \item Let $\delta_t(c)$ be the candidate minimum degree
    in $\Gfill{t}$ using $k=10c\lceil \log n \rceil$ sketches.
\item If $\delta_t(c) > c/2$, then set $c \leftarrow 2c$ and repeat.
\end{enumerate}

The core idea of the routine above is that if the candidate minimum degree is at
most $c/2$, then with high probability the true minimum degree is at most $c$.
It follows that using $O(c \log n)$ sketching data structures
guarantees the minimum degree estimate is correct with high probability. 

\begin{proof}[Proof of Theorem~\ref{thm:OutputSensitiveAlgo}.]
The proof is analogous to that of Theorem~\ref{thm:BoundedDegreeAlgo}.
The upper bound for the minimum degrees is now 
$\Delta = 2 \cdot \max_{t \in [n]} \delta_t$, and so
the time and space complexities follow. 
\end{proof}

\subsection{Computing an Approximate Minimum Degree}
\label{subsec:ApproxDegreeDS}
To avoid bounding the minimum fill degree over all steps and to make the
running time independent of the output, we modify the previous algorithms to
obtain an approximate min-degree vertex at each step.
We reduce the number of $\ell_0$-sketches
and use the reciprocal of the $(1-1/e)$-th order statistic to approximate the
cardinality $\variable{size\_of\_minimizers}[u]$ (and hence the 1-degree of $u$)
to obtain a nearly-linear time approximation algorithm.

There is, however, a subtle issue with the randomness involved with this algorithm.
A necessary condition for the algorithm to succeed as intended is that the
sketches at each step are independent of the past decisions of the algorithm.
Therefore, we must remove all dependencies between previous and current queries.
In Section~\ref{subsec:Correlation} we demonstrate how correlations between
steps can amplify.
To avoid this problem, we must decorrelate the current state of the sketches
from earlier pivoting updates to the data structures.
We carefully address this issue in Section~\ref{sec:Decorrelation}.
Instead of simply selecting a vertex with an approximate min-degree, this
algorithm instead requires access to all vertices whose estimated degree is
within a certain range of values.
Therefore, this approximation algorithm uses a bucketing data structure, as
opposed to the previous two versions that output the vertex to be pivoted.
Figure~\ref{fig:ApproxDegreeDSGlobalVar} describes the global data structures
for this version of the algorithm.

\begin{figure}[H]
  
  \begin{algbox}
   \underline{Global Variables}: graph $G$,
    error tolerance $\epsilon > 0$.
    
    \begin{enumerate}
      \item $k$, the number of sketches set to
        $50 \left\lceil \log{n} \epsilon^{-2} \right\rceil$.
 
      \item $k$ independent $\ell_0$-sketch data structures
      \[
        \variable{dynamic\_sketch}[1],
        \variable{dynamic\_sketch}[2],
        \dots,
        \variable{dynamic\_sketch}[k].
      \]

      \item For each vertex $u$, a balanced binary search
      tree $\variable{minimizers}[u]$ that stores 
      $\minimizer_i(u)$ across all $i \in [k]$ $\ell_0$-sketches,
      and maintains the element in $\variable{minimizers}[u]$ with rank
      \[
        \left\lfloor k \left(1- \dfrac{1}{e} \right)\right\rfloor.
      \]

      \item A balanced binary tree $\variable{quantile}$
      over all vertices $u$ whose key is the
      $\lfloor k \left(1-1/e\right)\rfloor$-ranked element in
        $\variable{minimizers}[u]$.
      
    \end{enumerate}
  \end{algbox}

\caption{Global variables and data structures for $\textsc{ApproxDegreeDS}$,
which returns (implicit) partitions of vertices into buckets with $\epsilon$-approximate degrees.}
\label{fig:ApproxDegreeDSGlobalVar}
\end{figure}

To successfully use fewer sketches, for a given vertex $u$ we estimate the
cardinality of the set of its minimizers via its order statistics
instead of using the exact cardinality as we did before with
the binary search tree $\variable{size\_of\_minimizers}[u]$.
Exploiting correlations in the order statistics of sketches is often the
underlying idea behind efficient cardinality estimation.
In particular, we make use of the following lemma, which is 
essentially a restatement of \cite[Propositions 7.1 and 7.2]{Cohen97}.

\begin{restatable}{lemma}{minValue}
\label{lem:minValue}
Suppose that we have $k$ copies of the $\ell_0$-sketch data structure,
for $k = 50 \left\lceil \log{n} \epsilon^{-2} \right\rceil$.
Let $u$ be any vertex such that $\deg(u)+1 > 2 \epsilon^{-1}$,
and let $Q(u)$ denote the
$\lfloor k \left(1-1/e\right)\rfloor$-ranked
key value in the list $\variable{minimizers}[u]$.
Then, with high probability, we have
\[
  \frac{1 - \epsilon}{\deg(u)+1}
     \leq Q(u) \leq \frac{1 + \epsilon}{\deg(u)+1}.
\]
\end{restatable}

\noindent
In \cite{Cohen97} they assume that the random keys $R(v)$ are drawn from the
exponential distribution (and hence the minimum key value is also), whereas
we assume that $R(v)$ is drawn independently from the uniform distribution.
When $\deg(u)$ is large enough though, the minimum of $\deg(u)$ random variables
from either distribution is almost identically distributed.
For completeness, we prove Lemma~\ref{lem:minValue} when the keys $R(v)$
are drawn from the uniform distribution in Appendix~\ref{sec:SketchingProofs}.

This idea leads to the following subroutine for providing implicit access to
all vertices with approximately the same degree. This is critical for our
nearly-linear time algorithm, and we explain its intricacies in
Section~\ref{sec:Decorrelation}. The pseudocode for this subroutine is given in
Figure~\ref{fig:approxdegree}.

\begin{figure}[H]
  \begin{algbox}
    $\textsc{ApproxDegreeDS\_Pivot}(u)$
    
    \underline{Input}: vertex to be pivoted, $u$.
    
    \underline{Output}: updated global state.
    
    \begin{enumerate}
      \item For each sketch $i=1$ to $k$:
         \begin{enumerate}
            \item $\left(v_1, v_2, \dots, v_{\ell} \right)
              \leftarrow \variable{dynamic\_sketch}[i].\textsc{PivotVertex}(u)$,
              the set of vertices in the $i$-th sketch whose minimizers changed
              after we pivot out $u$.
            \item For each $j=1$ to $\ell$:
              \begin{enumerate}
                \item Update the values corresponding to sketch $i$ in
                   $\variable{minimizers}[v_j]$,
                   which in turn updates its
                   $\lfloor k(1 - 1/e) \rfloor$-ranked quantile.
                \item Update the entry for $v_j$ in
                  $\variable{quantile}$ with the new value of the
                  $\lfloor k(1 - 1/e) \rfloor$-ranked quantile of
                  $\variable{minimizers}[v_j]$.
    \end{enumerate}
    \end{enumerate}
    \end{enumerate}
  \end{algbox}
  
  \begin{algbox}
    $\textsc{ApproxDegreeDS\_Report()}$
    
    \underline{Output}: approximate bucketing of the vertices
    by their fill 1-degrees.
    
    \begin{enumerate}
      \item For each $i=0$ to $B = O(\log{n} \epsilon^{-1})$:
      \begin{enumerate}
        \item Set $V_i$ to be the split binary tree
        in $\variable{quantile}$ that contains all
        nodes with $\lfloor k(1 - 1/e) \rfloor$-ranked quantiles in the range
      \[
        \left[\left(1 + \epsilon\right)^{-(i + 1)}, \left(1 + \epsilon\right)^{-i}\right].
      \]
      \end{enumerate}
     \item Return $(V_1, V_2, \dots, V_{B})$.
    \end{enumerate}
    
  \end{algbox}
  
  \caption{Pseudocode for the data structure that returns pointers to binary
  trees containing partitions of the remaining vertices into sets with
  $\epsilon$-approximate degrees.
  }
  \label{fig:approxdegree}
\end{figure}

Observe that because 1-degrees are bounded by $n$,
whenever we call $\textsc{ApproxDegreeDS\_Report}$ we have
$B = O(\log n \epsilon^{-1})$ with high probability by
Lemma~\ref{lem:minValue}.
Therefore, this data structure can simply return pointers
to the first element in each of the partitions $V_1, V_2, \dots, V_B$.

\begin{proof}[Proof of Theorem~\ref{thm:ApproxDegreeDS}.]
By construction,
all vertices in $V_i$ have their $\lfloor k(1 - 1/e) \rfloor$-ranked quantile
in the range 
\[
  \left[(1+\epsilon)^{-(i+1)},(1+\epsilon)^{-i}\right].
\]
By Lemma~\ref{lem:minValue}, the 1-degree of any vertex in bucket $V_i$
lies in the range
\[
  \left[(1-\epsilon)(1+\epsilon)^{i}, (1+\epsilon)^{i+2}\right]
\]
with high probability, which is within the claimed range for $\epsilon \leq 1/2$.

The proof of time and space complexities is similar to that of
Theorem~\ref{thm:BoundedDegreeAlgo}.
Letting the number of sketches $k = O(\log{n} \epsilon^{-2})$ instead of
$O(\Delta \log n)$ proves the space bound.
One of the main differences in this data structure is that
we need to store information about the
$\lfloor k(1 - 1/e) \rfloor$-ranked quantiles.
These queries can be supported in $O(\log{n})$ time by augmenting a balanced
binary search tree with information about sizes of the subtrees in standard
ways (e.g., \cite[Chapter 14]{CLRS3}).
It follows that the total cost of all calls to
\textsc{ApproxDegreeDS\_Pivot} is $O(m \log^3{n} \epsilon^{-2})$.
To analyze each call to \textsc{ApproxDegreeDS\_Report}, we use standard
splitting operations for binary search trees (e.g., treaps~\cite{Seidel96}),
which allows us to construct each bucket in $O(\log{n})$ time.
\end{proof}

\noindent
Note that there will be overlaps between the 1-degree intervals, so determining
which bucket contains a given vertex is ambiguous if its order statistic is
near the boundary of an interval.

An immediate corollary of Theorem~\ref{thm:ApproxDegreeDS} is that
we can provide access to approximate min-degree vertices
for a fixed sequence of updates by always returning
an entry from the first nonempty bucket.

\begin{corollary}
\label{cor:ApproxMindegreeOblivious}
For a fixed elimination ordering $(u_1, u_2, \dots, u_n)$,
we can find $(1 + \epsilon)$-approximate minimum degree vertices
in each of the intermediate states in 
$O(m \log^{3}n\epsilon^{-2})$ time.
\end{corollary}

\noindent
It is also possible to adaptively choose the number of sketches
for the $(1+\epsilon)$-approximate minimum degree algorithm by using a
subroutine that is similar to the one in Section~\ref{subsec:AdaptiveSketching}.

%% file: Decorrelation.tex

\section{Generating Decorrelated Sequences}
\label{sec:Decorrelation}

In this section we present a nearly-linear
$(1 + \epsilon)$-approximate marginal min-degree algorithm.
This algorithm relies on degree approximation via sketching,
as described in Theorem~\ref{thm:ApproxDegreeDS}.
In particular, it uses the randomized data structure $\variable{ApproxDegreeDS}$,
which provides access to buckets of vertices where the $i$-th bucket
contains vertices with fill 1-degree in the range
$[(1+\epsilon)^{i-2}, (1+\epsilon)^{i+2}]$.

\begin{theorem}
\label{thm:ApproxMinDegree}
There is an algorithm $\textsc{ApproxMinDegreeSequence}$
that produces a $(1+\epsilon)$-approximate marginal min-degree ordering
in expected $O(m \log^5 n \epsilon^{-2})$ time with high probability.
\end{theorem}

At each step of this algorithm, reporting any member of the first nonempty
bucket gives an approximate minimum degree vertex to pivot.  However, such
a choice must not have any dependence on the randomness used to get to this
step, and more importantly, it should not affect pivoting decisions in future
steps.  To address this issue, we introduce an additional layer of
randomization that decorrelates the $\ell_0$-sketches and the choice of
vertices to pivot.
Most of this section focuses on our technique for
efficiently decorrelating such sequences.

The pseudocode for \textsc{ApproxMinDegreeSequence} is given in
Figure~\ref{fig:ApproxMinDegreeSequence}. This algorithm makes use of the
following global data structures and subroutines.
\begin{itemize}
  \item $\variable{ApproxDegreeDS}$:
    Returns buckets of vertices with approximately equal 1-degrees
    (Section~\ref{subsec:ApproxDegreeDS}).
	\item \textsc{ExpDecayedCandidates}: Takes a sequence of values
    that are within $1 \pm \epsilon$ of each other,
	  randomly perturbs the elements, and returns the new ($\epsilon$-decayed) sequence
    (Section~\ref{subsec:epsilonDecayed}).
  \item \textsc{EstimateFill1Degree}:
    Gives an $\epsilon$-approximation
	  to the 1-degree of any vertex (Section~\ref{sec:DegreeEstimation}).
\end{itemize}

\noindent
We give the formal statement for \textsc{EstimateFill1Degree} in the following result.

\begin{theorem}\label{thm:DegreeEstimation}
There is a data structure that maintains a component graph $G^{\circ}$
under (adversarial) vertex pivots in a total of $O(m \log^2{n})$ time
and supports the operation $\textsc{EstimateFill1Degree}(u, \epsilon)$,
which given a vertex $u$ and
error threshold $\epsilon > 0$,
returns with high probability a $\epsilon$-approximation to the fill 1-degree of
$u$ by making $O(\deg(u) \log^2{n} \epsilon^{-2})$ oracle
queries to $G^{\circ}$.
\end{theorem}

\begin{figure}[H]
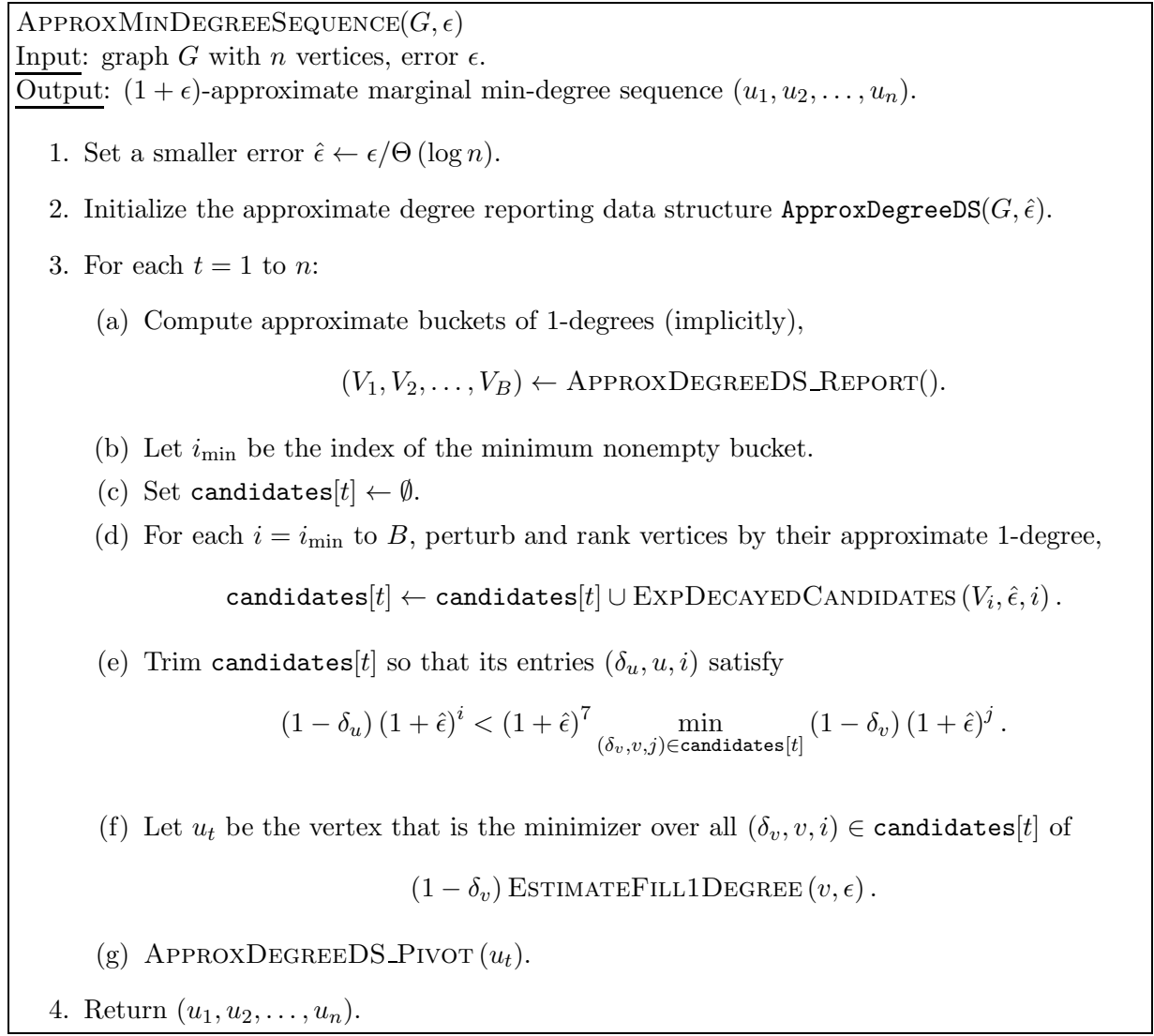

	
	\begin{algbox}
		
		$\textsc{ApproxMinDegreeSequence}(G, \epsilon)$
		
		\underline{Input}: graph $G$ with $n$ vertices, error $\epsilon$.
		
    \underline{Output}: $(1+\epsilon)$-approximate marginal min-degree
    sequence $(u_1, u_2, \dots, u_n)$.
		
		\begin{enumerate}
			
			\item Set a smaller error
      $\hat{\epsilon}
        \leftarrow \epsilon / \Theta\left(\log{n}\right)$. \label{algline:smallerror}
			
			\item 
			Initialize the approximate degree reporting data structure
      $\variable{ApproxDegreeDS}(G, \hat{\epsilon})$.
			
			\item For each $t=1$ to $n$:
			\begin{enumerate}
				
				\item Compute approximate buckets of 1-degrees (implicitly),
				\[
          \left( V_1 , V_2, \dots, V_B \right)
          \leftarrow
          \textsc{ApproxDegreeDS\_Report}().
				\]
				
				\item Let $i_{\min}$ be the index of the minimum nonempty bucket.
				
        \item Set $\variable{candidates}[t] \leftarrow \emptyset$.
				
				\item For each $i=i_{\min}$ to $B$,
        perturb and rank vertices by their approximate 1-degree,
				\[
          \variable{candidates}[t]
          \leftarrow \variable{candidates}[t]
              \cup \textsc{ExpDecayedCandidates}\left( V_{i}, \hat{\epsilon}, i \right).
				\]
				
      \item Trim $\variable{candidates}[t]$ so that its
				entries $(\delta_u, u, i)$ satisfy
				\[
          \left( 1 - \delta_{u} \right)
          \left( 1 + \hat{\epsilon} \right)^{i}
          <
          \left( 1 + \hat{\epsilon} \right)^{7}
          \min_{(\delta_v, v, j) \in \variable{candidates}[t]}
          \left( 1 - \delta_{v} \right)
          \left( 1 + \hat{\epsilon} \right)^{j}.
				\] \label{algline:trimGlobalCandidates}

        \item Let $u_t$ be the vertex that is the minimizer
          over all $(\delta_{v}, v, i) \in \variable{candidates}[t]$
        of
				\[
          \left( 1 - \delta_{v} \right)
          \textsc{EstimateFill1Degree}\left(v, \epsilon \right).
				\]

      \item $\textsc{ApproxDegreeDS\_Pivot}\left(u_t\right)$.
			\end{enumerate}
			
    \item Return $\left(u_1,u_2,\dots,u_n\right)$.
		\end{enumerate}
		
	\end{algbox}
	
  \caption{Pseudocode for the $(1+\epsilon)$-approximate marginal
  minimum degree ordering algorithm.} \label{fig:ApproxMinDegreeSequence}
\end{figure}

The most important part of this algorithm is arguably the use of exponential
random variables to construct a list of candidates that is completely uncorrelated
with the randomness used to generate the $\ell_0$-sketches and the choice of
previous vertex pivots.  The next subsection summarizes some desirable
properties of exponential distributions that we exploit for efficient
perturbations.

\subsection{Exponential Random Variables}
The exponential distribution is a continuous analog of the
geometric distribution that describes the time between events in
a Poisson point process.
We utilize well-known facts about its order statistics, which 
have also appeared in the study of fault tolerance and
distributed graph decompositions~\cite{MillerPX13}.
For a rate parameter $\lambda$, the exponential distribution
$\Exp(\lambda)$
is defined by the probability density function (PDF)
\begin{align*}
  f_{\Exp(\lambda)}(x) =
  \begin{cases}
    \lambda \exp\left(-\lambda x\right) & \text{if $x \ge 0$,}\\
    0 & \text{otherwise.}
  \end{cases}
\end{align*}
We will also make use of its cumulative density function (CDF)
\begin{align*}
  F_{\Exp(\lambda)}(x) &=
  \begin{cases}
    1 - \exp\left(-\lambda x\right) & \text{if $x \ge 0$,}\\
    0 & \text{otherwise}.
  \end{cases}
\end{align*}

A crucial property of the exponential distribution is that it is memoryless.
This means that for any rate $\lambda > 0$ and $s,t \ge 0$,
an exponentially distributed random variable $X$ satisfies the relation
\[
  \Pr\left[X > s + t \mid X > s\right] = \Pr\left[X > t\right].
\]
A substantial portion of our analysis relies on the \emph{order statistics} of
exponential random variables.
Given $n$ random variables $X_1, X_2, \dots, X_n$, the $i$-th
order statistic is  the value of the $i$-th minimum random variable.
A useful fact about i.i.d.\ exponential random variables is that the
difference between consecutive order statistics also follows an exponential
distribution. The algorithmic consequences of this property are that we can
sample the smallest (or largest) $k$ of $n$ exponential random variables in
increasing (or decreasing) order without ever generating all $n$ random
variables.

\begin{lemma}[{\cite{Feller71:book}}]
\label{lem:OrderStatisticExp}
Let $X_{(i)}^n$ denote the $i$-th order statistic of $n$
i.i.d.\ random variables drawn from the distribution $\Exp(\lambda)$.
Then, the $n$ variables $X_{(1)}^n, X_{(2)}^n- X_{(1)}^n,\dots, X_{(n)}^n - X_{(n - 1)}^n$
are independent, and the density of $X_{(k + 1)}^n - X_{(k)}^n$
is given by the distribution $\Exp((n - k)\lambda)$.
\end{lemma}

One approach to prove Lemma~\ref{lem:OrderStatisticExp} uses the i.i.d.\
assumption to show that the CDF of
$X_{(1)}^n$ is 
\begin{align*}
  F_{X^n_{(1)}}(x)&=1-(1-F_{\Exp(\lambda)}(x))^n
  \\&=1-\exp(-n\lambda x).
\end{align*}
This proves that $X^n_{(1)}$ follows an exponential
distribution with rate $n\lambda$.
Conditioning on $X^n_{(1)}$, we see that $X^n_{(2)}-X^n_{(1)}$
follows an exponential distribution equal to $X_{(1)}^{n-1}$
by the memoryless property. Therefore, one can repeat this argument to get the
density of $X^n_{(k+1)}-X^n_{(k)}$ for all $k$ up to $n-1$.

\subsection{Implicitly Sampling $\epsilon$-Decayed Minimums}
\label{subsec:epsilonDecayed}

The key idea in this section is the notion of $\epsilon$-decay, which
we use to slightly perturb approximate 1-degree sequences.
It is motivated by the need to decorrelate the list of vertices grouped
approximately by their 1-degree from previous sources of randomness in the
algorithm.
In the following definition, $n$ is the number of vertices in the original
graph before pivoting and $c_1 > 1$ is a constant.

\begin{definition}
\label{def:DecayedMinimum}
Given a sequence $(x_1,x_2,\dots,x_k) \in \mathbb{R}^{k}$,
we construct the corresponding \emph{$\epsilon$-decayed sequence}
$(y_1,y_2,\dots,y_k)$ by independently sampling
the exponential random variables
\[
  \delta_{i} \sim \hat{\epsilon} \cdot \Exp(1),
\]
where $\hat{\epsilon} = \epsilon/(c_1 \log{n})$
as in line~\ref{algline:smallerror} in \textsc{ApproxMinDegreeSequence},
and letting
\[
  y_i \leftarrow \left(1 - \delta_i\right)x_i.
\]
We say that the \emph{$\epsilon$-decayed minimum} of $(x_1,x_2,\dots,x_k)$
is the value $\min(y_1,y_2,\dots,y_k)$.
\end{definition}

\begin{definition}
\label{def:PerturbedMinDegreeSequence}
Given an error parameter $\epsilon > 0$ and an $\epsilon$-approximate
1-degree estimation routine $\textsc{Estimate1Degree}(G,u)$,
an \emph{$\epsilon$-decayed minimum degree ordering} is a sequence
such that:
\begin{enumerate}
  \item The vertex $u_t$ corresponds to the $\epsilon$-decayed minimum of
   $\textsc{Estimate1Degree}\left(\Gfill{t-1}, v\right)$
  over all remaining vertices $v \in \Vfill{t-1}$.
  \item The fill graph $\Gfill{t}$ is obtained after eliminating
    $u_t$ from $\Gfill{t-1}$.
\end{enumerate}
\end{definition}

Observe that the randomness of this perturbed degree estimator is regenerated
at each step and thus removes any previous dependence.  Next, we show that
this adjustment is a well-behaved approximation, and then we show how to
efficiently sample an $\epsilon$-decayed minimum degree.

\begin{lemma}
\label{lem:ApproxFactor}
Let $Y$ be an $\epsilon$-decayed minimum
of $(x_1, x_2, \dots, x_k)$.
With high probability, we have
\[
  Y \ge (1 - \epsilon) \min\left(x_1, x_2, \dots, x_k\right).
\]
\end{lemma}

\begin{proof}
We bound the probability of the complementary event
\[
  \Pr\left[Y < (1 - \epsilon ) \min\left(x_1, x_2, \dots, x_k\right) \right]. 
\]
Observe that we can upper bound this probability by
the probability that some $x_i$ decreases
to less than $1 - \epsilon$ times its original value.
Recall that we set $\hat{\epsilon} = \epsilon/(c_1 \log{n})$
for some constant $c_1 > 1$.
Consider $k$ i.i.d.\ exponential random variables
$X_1,X_2,\dots,X_k \sim \Exp(1)$,
and let
\[
  \delta_i = \hat{\epsilon} \cdot X_i,
\]
as in the definition of an $\epsilon$-decayed minimum.
Using the CDF of the exponential distribution, for each
$i \in [k]$ we have
\begin{align*}
  \Pr\left[\delta_{i} > \epsilon\right]
  &=
  \Pr\left[\frac{\epsilon}{c_1 \log n} \cdot X_i > \epsilon\right]\\
  &= \Pr\left[X_i > c_1 \log n\right]\\
  &= \exp \left( - c_1 \log{n} \right).
\end{align*}
It follows by a union bound that
\begin{align*}
  \Pr\left[
    \max_{i \in [k]}
     \delta_{i} >  \epsilon \right]
  &\leq \sum_{i=1}^n
  \Pr\left[\delta_{i} > \epsilon \right] \\
  &= n^{1-c_1},
\end{align*}
which completes the proof since $c_1 > 1$.
\end{proof}

By the previous lemma,
to produce a $(1+\epsilon)$-approximate marginal minimum degree
ordering, it suffices to compute an $\epsilon$-decayed minimum degree ordering.
Specifically, at each step we only need to
find the $\epsilon$-decayed minimum among the approximate fill 1-degrees of the
remaining vertices.  It turns out, however, that computing the approximate
1-degree for each remaining vertex in every iteration is expensive,
so we avoid this problem by using 
\textsc{ExpDecayedCandidates} on each bucket of vertices
to carefully select a
representative subset of candidates,
and then we pivot out the minimizer over all buckets.
The pseudocode for this subroutine is given in
Figure~\ref{fig:ExpDecayedCandidates}, where we again let $\hat{\epsilon} =
\epsilon / (c_1 \log{n})$ for some constant $c_1 > 1$.
Next, we show that this sampling technique is equivalent to finding the
$\epsilon$-decayed minimum over all remaining vertices with high probability.

\begin{figure}[H]
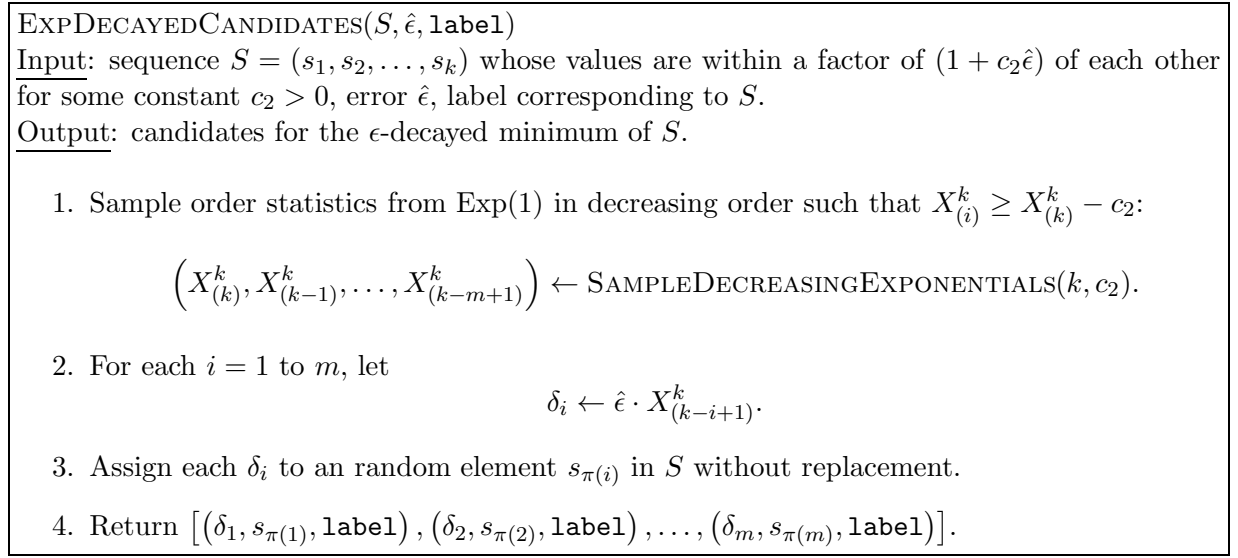


\begin{algbox}
$\textsc{ExpDecayedCandidates}
   (S, \hat{\epsilon},\variable{label})$

\underline{Input}:
sequence $S=(s_1,s_2,\dots,s_k)$ 
whose values are within a factor of $(1 + c_2 \hat{\epsilon})$ of each other
for some constant $c_2>0$, error $\hat{\epsilon}$, label corresponding to $S$.

\underline{Output}:
candidates for the $\epsilon$-decayed minimum of $S$.

\begin{enumerate}

\item 
Sample order statistics from $\Exp(1)$ in decreasing order
such that $X_{(i)}^k \ge X_{(k)}^k - c_2$:
\[
  \left( X_{(k)}^k, X_{(k-1)}^k, \dots, X_{(k - m + 1)}^k \right) \leftarrow
    \textsc{SampleDecreasingExponentials}(k, c_2).
\]

\item For each $i = 1$ to $m$, let
\[
  \delta_{i} \leftarrow \hat{\epsilon} \cdot X_{(k - i + 1)}^k.
\]

\item Assign each $\delta_{i}$ to an random element $s_{\pi(i)}$
in $S$ without replacement.

\item Return 
  $\left[
    \left(\delta_{1}, s_{\pi\left(1\right)}, \variable{label}\right),
    \left(\delta_{2}, s_{\pi\left(2\right)}, \variable{label}\right),
    \dots,
    \left(\delta_{m}, s_{\pi\left(m\right)}, \variable{label}\right)
    \right]$.
\end{enumerate}

\end{algbox}

\caption{
Pseudocode for generating an expected constant-size list of candidates for the
$\epsilon$-decayed minimum of a sequence of values that are within
$(1 + c_2\hat{\epsilon})$ of each other.
}

\label{fig:ExpDecayedCandidates}
\end{figure}

Note that the input sequence to \textsc{ExpDecayedCandidates}
requires that all its elements are within a factor of
$(1 + c_2\hat{\epsilon})$ of each other.
We achieve this easily using the vertex buckets returned by
$\textsc{ApproxDegreeDS\_Report}$ in Section~\ref{subsec:ApproxDegreeDS}
when $\hat{\epsilon}$ is the error tolerance.
The next lemma shows that the approximate vertex 1-degrees in
any such bucket satisfy the required input condition.

\begin{lemma}\label{lem:BucketBounds}
For any bucket $V_i$ of vertices returned by
$\textsc{ApproxDegreeDS\_Report}$, there exists a constant $c_2>0$ such that
all of the approximate 1-degrees are within a factor
of $(1+c_2 \hat{\epsilon})$ of each other.
Alternatively, all of the approximate 1-degrees are within a factor
of $(1+\hat{\epsilon})^7$ of each other.
\end{lemma}

\begin{proof}
The bucket $V_i$ has vertices with approximate 1-degrees in the range
\[
  \left[(1+\hat{\epsilon})^{i-2},(1+\hat{\epsilon})^{i+2}\right].
\]
by Theorem~\ref{thm:ApproxDegreeDS}.
We have oracle access to the component graph $\Gcomp{t}$ by
Theorem~\ref{thm:DegreeEstimation} and therefore can invoke
\textsc{EstimateFill1Degree} on it.
Instead of treating calls to $\textsc{EstimateFill1Degree}$
and the values of $\delta_i$ used to generate $\epsilon$-decayed
minimums as random variables, we view them as fixed values
by removing the bad cases with high probability.
That is, we define
\begin{equation*}
\label{eqn:EstimateFill1DegreeDef}
  \widetilde{\degfill{t}}(u)+1 \defeq
    \textsc{EstimateFill1Degree}\left(u,\hat{\epsilon}\right).
\end{equation*}
Every call to $\textsc{EstimateFill1Degree}$ is correct with high probability
by Theorem~\ref{thm:DegreeEstimation}, so we have
\[
  \left(1-\hat{\epsilon}\right)\left(\degfill{t}(u)+1\right)
     \le \widetilde{\degfill{t}}(u)+1
     \le \left(1+\hat{\epsilon}\right)\left(\degfill{t}(u)+1\right).
\]
This implies that all the approximate 1-degrees in bucket $V_i$ are in the
range
\[
  \left[\left(1+\hat{\epsilon}\right)^{i-4}, \left(1+\hat{\epsilon}\right)^{i+3}\right],
\]
for $\hat{\epsilon}$ sufficiently small.
Therefore, all of these values are within a factor
\[
  \frac{\left(1+\hat{\epsilon}\right)^{i+3}}{\left(1+\hat{\epsilon}\right)^{i-4}}
   = \left(1+\hat{\epsilon}\right)^{7}
   \le \left(1+c_2 \hat{\epsilon}\right)
\]
of each other for some $c_2 > 0$, which completes the proof.
\end{proof}

The most important part of 
\textsc{ExpDecayedCandidates} is
generating the order statistics efficiently.
In Figure~\ref{fig:SamplingDecreasingExponentials}
we show how to iteratively sample the variables
$X_{(i)}^k$ in decreasing order
using Lemma~\ref{lem:OrderStatisticExp}.
This technique is critically important for us because we only consider order
statistics satisfying the condition $X_{(k)}^k - c_2$, which is at most
a constant number of variables in expectation.

\begin{figure}[H]
\begin{algbox}
\textsc{SampleDecreasingExponentials}$(k, c_2)$

\underline{Input}:
integer $k \ge 0$, real-valued threshold $c_2 > 0$.

\underline{Output}:
order statistics $X_{(k)}^k, X_{(k-1)}^k, \dots, X_{(k-m+1)}^k$
from $\Exp(1)$
such that $X_{(k-m+1)}^k \ge X_{(k)}^k - c_2$
and $X_{(k-m)}^k < X_{(k)}^k - c_2$.

\begin{enumerate}

\item Sample $X_{(k)}^k$ using its CDF
\[
  \Pr\left[X_{(k)}^k \le x\right] = \left(1 - e^{-x}\right)^k.
\]

\item For each $i=1$ to $k-1$:
\begin{enumerate}
  \item Sample the difference $Y \sim \Exp(i)$ and let
    \[
      X_{(k-i)}^k \leftarrow X_{(k - i + 1)}^k - Y.
    \]
  \item If $X_{(k-i)}^k < X_{(k)}^k - c_2$,
    let $m \leftarrow i$ and
    exit the loop.
\end{enumerate}
  \item Return 
    $\left(X_{(k)}^k, X_{(k-1)}^k, \dots, X_{(k-m+1)}^k\right)$.
\end{enumerate}

\end{algbox}

\caption{
Pseudocode for iteratively generating order statistics of exponential random variables
in decreasing order within a threshold $c_2$ of the maximum value $X_{(k)}^k$.
}

\label{fig:SamplingDecreasingExponentials}

\end{figure}

To show that our algorithm is correct, we must prove that (1) the algorithm
selects a bounded number of candidates in expectation at each step, and (2) the
true $\epsilon$-decayed minimum belongs to the candidate list.
We analyze both of these conditions in the following lemma.

\begin{lemma}
\label{lem:FindingDecayedMin}
If $x_1, x_2, \dots, x_k$ are within a factor of
$(1 + c_2\hat{\epsilon})$ of each other,
then the $\epsilon$-decayed minimum 
is among the candidates returned by
$\textsc{ExpDecayedCandidates}((x_1, x_2, \dots, x_k), \hat{\epsilon},\cdot)$.
Furthermore, the expected number of candidates returned is bounded by the
constant $e^{c_2}$.
\end{lemma}

\begin{proof}
Let $X_1,X_2, \dots, X_k \sim \Exp(1)$ be i.i.d.\
and let the order statistic $X_{(k)}^{k} = \max\{X_1,X_2,\dots,X_k\}$.
We first verify the correctness of
\textsc{SampleDecreasingExponentials}.
Using the CDF of the exponential distribution, it follows that the
CDF of $X_{(k)}^k$ is
\begin{align*}
  F_{X_{(k)}^k}(x) &= \Pr\left[\max\{X_1,X_2,\dots,X_n\} \le x\right]\\
  &= \prod_{i=1}^k \Pr\left[X_i \le x\right]\\
  &= \left(1 - e^{-x}\right)^k.
\end{align*}
The memoryless property implies that
we can generate
$X_{(k-1)}^k, X_{(k-2)}^k, \dots, X_{(1)}^k$ iteratively
by sampling their differences from
an exponential distribution whose rate is given by
Lemma~\ref{lem:OrderStatisticExp}.

Now we show that the $\epsilon$-decayed minimum is among the candidates.
We claim that it suffices to sample every $X_{(i)}^k$ such that
$X_{(i)}^k \ge X_{(k)}^k - c_2$, or equivalently every $\delta_i$ such that
$\delta_i \ge \delta_k - c_2 \hat{\epsilon}$.
To see this, suppose for contradiction that the $\epsilon$-decayed minimum
$x_{\pi(j)}$ is not included.
Then we have $\delta_{j} < \delta_k - c_2 \hat{\epsilon}$, so it follows that
\begin{align*}
  \left(1-\delta_{j}\right) x_{\pi(j)}
   &> \left(1 - \delta_k + c_2\hat{\epsilon}\right) \frac{x_{\pi(k)}}{1+c_2 \hat{\epsilon}}
   \ge (1-\delta_k)x_{\pi(k)},
\end{align*}
which is a contradiction.
Therefore, our candidate list contains the $\epsilon$-decayed minimum.

Lastly, to count the expected number of candidates, we count the number of
values $\delta_j$ generated.
Let $Z_i$ be the indicator variable for the event $X_i \ge X_{(k)}^k - c_2$.
Then $Z = \sum_{i=1}^k Z_i$ indicates the size of the candidate list.
Using the memoryless property of exponential random variables, we have
\begin{align*}
  \E[Z] &= \sum_{i=1}^k \E[Z_i]\\
     &= \sum_{i=1}^k \prob{}{X_{(i)}^k \ge X_{(k)}^k - c_2} \\
     &= 1 + \sum_{i=1}^{k-1} \prob{}{X_{(i)}^k \ge X_{(k)}^k - c_2} \\
     &= 1 + \sum_{i=1}^{k-1} \prob{}{X_{(i)}^i \le c_2}\\
     &= 1 + \sum_{i=1}^{k-1} \left( 1 - e^{-c_2} \right)^i\\
     &\le e^{c_2},
\end{align*}
where the final equality considers the geometric series.
Therefore, at most a constant number of exponential random variables are
generated as we sample backwards from the maximum.
\end{proof}

We cannot simply work with the first nonempty bucket because the randomness
introduces a $1\pm\epsilon$ peturbation.  Furthermore, the bucket
containing the vertex with minimum degree is dependent on the randomness of
the sketches (as discussed in Theorem~\ref{thm:ApproxDegreeDS}).
To bypass this problem we inject additional, uncorrelated randomness into
the algorithm at each step to find $O(1)$ candidates for each of
the $O(\log n \hat{\epsilon}^{-1})$ buckets, which
increases the number of global candidates to
$O(\log n \hat{\epsilon}^{-1})$.
Then in the penultimate step of each iteration, before we compute the
approximate 1-degrees of candidate vertices (which is somewhat expensive), we
carefully filter the global list so that the global $\epsilon$-decayed minimum
remains in the list with high probability.

\begin{lemma}\label{lem:trimmingisokay}
Let $(\delta_{u},u,i)$ be the entry over all $(\delta_{v},v,j) \in
\variable{candidates}[t]$ that minimizes
\[
  \left(1 - \delta_{v}\right) \textsc{EstimateFill1Degree}(v,\epsilon).
\]
Then, with high probability, we have
\[
  \left( 1 - \delta_{u} \right)
  \left( 1 + \hat{\epsilon} \right)^{i}
  \leq
  \left( 1 + \hat{\epsilon} \right)^{7}
  \min_{(\delta_v, v, j) \in \variable{candidates}[t]}
  \left( 1 - \delta_{v} \right)
  \left( 1 + \hat{\epsilon} \right)^{j}.
\]
\end{lemma}

\begin{proof}
Let $(\delta_{v},v,j)$ be an arbitrary entry in
$\variable{candidates}[t]$.
By assumption, we have
\[
  \left( 1 - \delta_{u} \right) \textsc{EstimateFill1Degree}(u,\epsilon)
  \leq \left( 1 - \delta_{v} \right) \textsc{EstimateFill1Degree}(v,\epsilon).
\]
Using inequalities in Lemma~\ref{lem:BucketBounds}, it follows that
\[
  \textsc{EstimateFill1Degree}(u,\epsilon) \geq (1+ \hat{\epsilon})^{i-4}
\]
and
\[
  \textsc{EstimateFill1Degree}(v,\epsilon) \le (1+ \hat{\epsilon})^{j+3}
\]
with high probability.
Substituting these into the previous inequality gives us the result.
\end{proof}

\subsection{Analysis of the Approximation Algorithm}
\label{subsec:AnalyzingApproxMinDeg}

Now that we have all of the building blocks from the previous subsection,
we prove the correctness of the $(1+\epsilon)$-approximate marginal
minimum degree algorithm and bound its running time.

\begin{lemma}
\label{lem:Correctness}
For any graph $G$ and any error $\epsilon$,
the output of $\textsc{ApproxMinDegreeSequence}(G, \epsilon)$
is a $(1+\epsilon)$-approximate marginal minimum degree sequence
with high probability.
\end{lemma}

\begin{proof}
We prove by induction that for some constant $c > 1$, after $t$ steps of the
algorithm, the output is a $(1+\epsilon)$-approximate marginal min-degree
ordering with probability at least $1 - tn^{-c}$.
The base case when $t = 0$ follows trivially because nothing has happened.
For the inductive hypothesis, assume that after $t$ steps the sequence
$(u_1, u_2, \dots, u_t)$ is a
$(1+\epsilon)$-approximate min-degree ordering
and let the graph state be $\Gfill{t}$, where the eliminated vertices
$u_1,u_2,\dots,u_t$ have been pivoted.

By Lemma~\ref{lem:BucketBounds}, all values in a given bucket are within
a factor of $1+c_2\hat{\epsilon}$ of each other.
We use the guarantees of Lemma~\ref{lem:FindingDecayedMin}
to compute the $\epsilon$-decayed minimum candidate of each bucket.  It follows
from Lemma~\ref{lem:trimmingisokay} that after we trim the candidate
list, one of the remaining candidates is the original minimizer of
$(1-\delta_v)\textsc{EstimateFill1Degree}(v,\epsilon)$ with high probability.
Therefore,~$u_t$ is the $\epsilon$-decayed minimum over all values of
$\textsc{EstimateFill1Degree}(v,\epsilon)$ with high probability.
Lastly, invoking the bound on distortions incurred by $\epsilon$-decay
in Lemma~\ref{lem:ApproxFactor} and accounting for the error of
\textsc{EstimateFill1Degree}, the 1-degree of $u_t$ is within $1+\epsilon$ of
the minimum 1-degree in the fill graph $\Gfill{t}$ with high probability.
Taking a union bound over all the high probability claims, we have a failure
probability of at most $n^{-c}$.
Thus, the inductive hypothesis also holds for $t+1$.
\end{proof}

We now analyze the cost of the algorithm.  To do this, we first show that if a
vertex is close to the global $\epsilon$-decayed minimum, then there is a
reasonable chance that it actually is the minimizer.  In other words, if the
algorithm queries the approximate degree of a vertex, then it is likely that
this vertex belongs to the $\epsilon$-decayed approximate degree sequence.
This explains the trimming condition in \textsc{ApproxMinDegreeSequence}.

\begin{lemma}
\label{lem:CandidateChosen}
For any constant $c_3 \ge 1$, choice of error $\hat{\epsilon}$,
sequence of values $(x_1,x_2, \dots, x_k)$,
and index $i \in [k]$, we have
\begin{align*}
  & \Pr\left[
  \normalfont{\text{$i$ corresponds to the $\epsilon$-decayed minimum
          of $(x_1, x_2, \dots, x_k)$}}\right]\\
  &\hspace{6.5cm} \geq \exp\left( -2 c_3 \right)
  \Pr\left[
\left( 1 - \delta_{i} \right) x_{i}
  < \left( 1 +  \hat{\epsilon}\right)^{c_{3}}
  \min_{j \in [k]} \left( 1 - \delta_{j} \right) x_{j}\right].
\end{align*}
\end{lemma}

\begin{proof}
Without loss of generality, suppose we generate $\delta_{i}$ last. Let
\[
  m = \min_{j \in[k]\setminus\{i\}} \left(1 - \delta_{j} \right) x_j
\]
be the previous $\epsilon$-decayed minimum.
If $m \geq x_i$ then both probabilities in the claim are equal to $1$ and the
result holds trivially.
Otherwise, consider the probability that $i$ corresponds to the minimizer
conditioned on the event
\[
\left(1 - \delta_{i} \right) x_i
< \left( 1 +  \hat{\epsilon} \right)^{c_{3}} m.
\]
Equivalently, assume that
$\delta_{i}
  >
  \gamma,$
for some $\gamma$ such that
$(1 -  \gamma ) x_i = \left( 1 +  \hat{\epsilon} \right)^{c_{3}} m$.
By the memoryless property of the exponential distribution,
we have
\begin{align*}
  \Pr\left[\delta_i > \gamma + 2c_3\hat{\epsilon} \mid \delta_i > \gamma\right] &=
  \Pr\left[\delta_i > 2c_3\hat{\epsilon}\right] \\
  &=\exp\left(-2c_3\right).
\end{align*}
Therefore, with probability at least $\exp(-2c_3)$, it follows that
\begin{align*}
  \left(1-\delta_i\right)x_i &< \left(1 - \gamma - 2c_3 \hat{\epsilon}\right)x_i\\
  &\le \left(1-2c_3 \hat{\epsilon}\right)\left(1-\gamma\right)x_i\\
  &= \left(1-2c_3\hat{\epsilon}\right)\left(1+\hat{\epsilon}\right)^{c_3} m\\
  &\le m.
\end{align*}
This means that if the decayed value of $x_i$ is within
a threshold of the previous minimum $m$, then~$x_i$ itself
will decay below $m$ with at least constant probability. 
\end{proof}

Making the substitution $c_3 = 7$,
as in line~\ref{algline:trimGlobalCandidates} of
\textsc{ApproxMinDegreeSequence}
in Figure~\ref{fig:ApproxMinDegreeSequence},
gives the following corollary,
which allows us to prove our main result.

\begin{corollary}\label{corollary:trimGlobalCandidates}
If a vertex $v$ is in $\variable{candidates}[t]$
after line~\ref{algline:trimGlobalCandidates} of \textsc{ApproxMinDegreeSequence},
then with probability at least $\exp(-14)$,
$v$ is the $\epsilon$-decayed minimum.
\end{corollary}

\begin{proof}[Proof of Theorem~\ref{thm:ApproxMinDegree}.]
The correctness follows from Lemma~\ref{lem:Correctness}.
We can maintain access to all of the buckets across the sequence of pivots
in a total time of
\[
  O\left(m \log^{3}n \hat{\epsilon}^{-2} \right)
  =
  O\left(m \log^{5}n \epsilon^{-2} \right)
\]
by Theorem~\ref{thm:ApproxDegreeDS}, so all that remains is bounding the total
cost of the calls to $\textsc{EstimateFill1Degree}$.

By Theorem~\ref{thm:DegreeEstimation},
the cost of maintaining the component graph under pivots is $O(m \log^2{n})$,
a lower order term.
To analyze the aggregate cost of the calls to $\textsc{EstimateFill1Degree}$,
we utilize Corollary~\ref{corollary:trimGlobalCandidates}, which states that
any vertex in $\variable{candidates}[t]$ (after trimming) is the one we pivot
with constant probability.
Specifically, we prove by induction on the number of remaining vertices
that for some constant $c_4$
the expected cost of calling $\textsc{EstimateFill1Degree}$
is bounded by
\[
  c_4 
  \left( 
  \sum_{u \in \Vfill{t} }
  \degrem{,t}\left( u \right)
  \right)
  \log^{2}{n} \hat{\epsilon}^{-2}.
\]

\noindent
The base case $t = n$ follows trivially since no vertices remain.
Now suppose that the claim is true for $t+1$ vertices.
By the induction hypothesis,
the expected cost of the future steps is bounded by
\begin{align*}
  &\sum_{u \in \Vfill{t}} \Pr\left[\text{$u$ is the $\epsilon$-decayed minimum}\right] \cdot
c_4 
\left( 
  -\degrem{,t}(u) +
  \sum_{v \in \Vfill{t} }
  \degrem{,t}(v) 
\right)
\log^{2}{n} \hat{\epsilon}^{-2}\\
  &\hspace{0.0cm}=
c_4 
  \left( 
    \sum_{u \in \Vfill{t}} \degrem{,t}(u)
  \right)
  \log^{2}{n} \hat{\epsilon}^{-2}
 -
c_4 
  \sum_{u \in \Vfill{t}} \Pr\left[\text{$u$ is the $\epsilon$-decayed minimum}\right]
  \cdot
  \degrem{,t}(u)
\log^{2}{n} \hat{\epsilon}^{-2}.
\end{align*}

\noindent
Now we consider the cost of evaluating $\textsc{EstimateFill1Degree}(u,
\hat{\epsilon})$ at time $t$ if $u \in \variable{candidates}[t]$.
By Corollary~\ref{corollary:trimGlobalCandidates} and expanding
the conditional probability, we have
\[
  \Pr\left[\text{$u$ is the $\epsilon$-decayed minimum}\right]
  \ge \exp(-14)
  \Pr\left[u \in \variable{candidates}[t]\right].
\]
Therefore, using Theorem~\ref{thm:DegreeEstimation}, the expected cost of these calls is
\begin{align*}
  &\sum_{u \in \Vfill{t}}
  \prob{}{u \in \variable{candidates}[t]} \cdot 
  c_3  \degrem{,t}(u)
\log^{2}{n} \hat{\epsilon}^{-2}\\
  &\hspace{4.0cm}\leq c_3 \exp(14) 
  \sum_{u \in \Vfill{t}} \prob{}{\text{$u$ is the $\epsilon$-decayed minimum}}
\cdot
  \degrem{,t}(u)
\log^{2}{n} \hat{\epsilon}^{-2}.
\end{align*}
It follows that the inductive hyptothesis holds
for all $t$ by letting $c_4 = c_3 \exp(14)$.
The initial sum of the remaining degrees is $O(m)$,
so the total expected cost of calling \textsc{EstimateFill1Degree} is
\[
  O\left( m \log^{2}{n} \hat{\epsilon}^{-2} \right)
  =
  O\left( m \log^{4}{n} \epsilon^{-2} \right),
\]
which completes the proof.
\end{proof}

%% file: DegreeEstimation.tex
\section{Estimating the Fill 1-Degree of a Vertex}
\label{sec:DegreeEstimation}

This section discusses routines for approximating the fill 1-degree of a vertex
in a partially eliminated graph.
We also show how to maintain the partially eliminated
graph throughout the course of the algorithm, which
allows us to prove Theorem~\ref{thm:DegreeEstimation}.
The partially eliminated graph we use for degree estimation is
the component graph $\Gcomp{}$, where connected components of the
eliminated vertices are contracted into single vertices called \emph{component}
vertices.
See Section~\ref{subsec:PreliminariesGaussian} for a detailed explanation.

Our goal is to efficiently approximate the fill 1-degree of a given remaining
vertex $u$. By the definition of fill 1-degree and the neighborhoods
of component graphs, it follows that
\[
  \degfill{}(u) + 1 = \left|\{u\}\cup\Nrem{}(u) \cup \bigcup_{x \in \Ncomp{}(u)} \Nrem{}(x) \right|. 
\]
In other words, the fill 1-neighborhood of $u$ is set of remaining 1-neighbors
of $u$ in the original graph in addition to the remaining neighbors of
each component neighbor of $u$.

This union-of-sets structure has a natural $(0,1)$-matrix interpretation, where
columns correspond to remaining vertices and rows correspond to neighboring
component neighborhoods of $u$ (along with an additional row for the
$1$-neighborhood of $u$).
For each row $i$, set the entry $A(i,j)=1$ if vertex~$j$ is in the $i$-th
neighborhood set and let $A(i,j)=0$ otherwise.
The problem can then be viewed as querying for the number of nonzero
columns of $A$.
Specifically, we show how one can accurately estimate fill 1-degrees using the
following matrix queries:
\begin{itemize}
  \item $\textsc{RowSize}(A,i)$: Return the number of nonzero elements in row $i$ of $A$.
  \item $\textsc{SampleFromRow}(A,i)$: Returns a column index $j$ uniformly at
    random from the nonzero entries of row $i$ of $A$.
  \item $\textsc{QueryValue}(A,i,j)$: Returns the value of $A(i,j)$.
\end{itemize}

\noindent
The main result in this section is the follow matrix sampler.

\NonZeroColumnEstimator

Before analyzing this matrix-based estimator, we
verify that Lemma~\ref{lem:NonZeroColumnEstimator} can be used
in the graph-theoretic setting
to prove Theorem~\ref{thm:DegreeEstimation}.
We use the following tools for querying degrees and sampling neighbors in a
component graph as it undergoes pivots.

\begin{restatable}[]{lemma}{DegreeEstimationDS}
\label{lem:DegreeEstimationDS}
We can maintain a component graph under
vertex pivots in a total time of $O(m \log^2{n})$.
Additionally, this component graph data structure grants $O(\log{n})$ time oracle access for:
\begin{itemize}
\item Querying the state of a vertex.
\item Querying the component or remaining neighborhood (and hence degree) of a vertex.
\item Uniformly sampling a remaining neighbor of a component or remaining
  vertex.
\item Uniformly sampling a random component vertex.
\end{itemize}
\end{restatable}

\noindent
We defer the proof of Lemma~\ref{lem:DegreeEstimationDS} to
Section~\ref{sec:DynamicGraphs}.

Assuming the correctness of Lemma~\ref{lem:NonZeroColumnEstimator}
and Lemma~\ref{lem:DegreeEstimationDS}, we can easily
prove Theorem~\ref{thm:DegreeEstimation}, which allows us to efficiently
estimate the fill 1-degrees of vertices throughout the algorithm.

\begin{proof}[Proof of Theorem~\ref{thm:DegreeEstimation}]
We can implicitly construct $A$ and simulate the matrix operations as follows.
Storing the adjacency list of the component graph using binary search trees
clearly implies that we can implement \textsc{RowSize} and \textsc{SampleFromRow}.
Moreover, in this setting \textsc{QueryValue} corresponds to querying
connectivity, which again is simple because we use binary search trees.
Substituting in the runtime bounds gives the desired result.
\end{proof}

The rest of the section is outlined as followed.
We prove a weaker but relevant version of
the matrix estimator (Lemma~\ref{lem:NonZeroColumnEstimator})
in Section~\ref{subsec:DegreeEstimation_Matrix}.
This algorithm relies on a subroutine to estimate the mean of a
distribution, which we discuss in Section~\ref{subsec:MeanEstimation}.
Then by more carefully analyzing the previous two algorithms, we prove the
original estimation result in Section~\ref{subsec:DegreeEstimation_Better}.

\subsection{Approximating the Number of Nonzero Columns using Mean Estimators}
\label{subsec:DegreeEstimation_Matrix}

We begin by defining an estimator for counting the number of nonzero
columns of $A$.
Let
\[
  \textsc{ColumnSum}(A,j)
  \defeq
  \sum_{i=1}^r A(i, j),
\]
and consider the normalized matrix $B$ such that
\begin{align*}
  B(i,j) = \begin{cases}
    0 & \text{if $A(i,j)=0$} \\
    \textsc{ColumnSum}(A,j)^{-1} & \text{if $A(i,j) = 1$}.
  \end{cases}
\end{align*}
If we know the number of nonzeros in $B$ and can uniformly sample nonzero
entries of $B$, then we can use this distribution on $B$ as an unbiased
estimator for the number of nonzero columns in $A$.
We explicitly capture this idea with the following lemma.

\begin{lemma}
\label{lem:EstimatorCorrectness}
If $X$ is a uniformly random nonzero entry of $B$,
then
\[
  \E\left[X\right] = \frac{\textsc{NonzeroColumns}(A)}{\nnz(A)}.
\]
\end{lemma}

\noindent
Assuming that we can uniformly sample indices $(i,j)$ of nonzero entries of
$B$, Lemma~\ref{lem:EstimatorCorrectness} implies that it is sufficient to
estimate column sums. We show how to do this by estimating the mean of an
appropriately chosen Bernoulli distribution on the column. All of
our estimators in this section use a general-purpose
\textsc{EstimateMean} algorithm (given in Figure~\ref{fig:EstimateMean}),
for any distribution over $[0,1]$.
We present its accuracy guarantees and sample complexity next,
and defer the proof to Section~\ref{lem:MeanEstimation}.

\begin{figure}[H]
	\begin{algbox}
		
		$\textsc{EstimateMean}(D, \sigma)$
		
    \underline{Input}: access to a distribution $D$ over $[0,1]$,
		cutoff threshold $\sigma > 0$.
		
		\underline{Output}: estimation of the mean of $D$.
		
		\begin{enumerate}
			\item Initialize $\variable{counter} \leftarrow 0$ and $\variable{sum} \leftarrow 0$.
			\item While {$\variable{sum} < \sigma$}:
			\begin{enumerate}
				\item Generate $X \sim D$.
				\item $\variable{sum} \leftarrow \variable{sum} + X$.
				\item $\variable{counter} \leftarrow \variable{counter} + 1$.
			\end{enumerate}
			\item Return $\sigma/\variable{counter}$.
		\end{enumerate}
		
	\end{algbox}
	
  \caption{Pseudocode for an algorithm that estimates the mean of a distribution $D$ on $[0,1]$.}
	\label{fig:EstimateMean}
\end{figure}

\begin{lemma}
\label{lem:MeanEstimation}
Let $D$ be any distribution over $[0,1]$ with (an unknown) mean $\mu$.
For any cutoff $\sigma > 0$ and error $\epsilon > 0$,
with probability at least $1 - \exp(-\epsilon^2 \sigma / 5)$,
the algorithm $\textsc{EstimateMean}(D,\sigma)$:
\begin{itemize}
  \item
  \label{part:RunTime}
  Generates $O(\sigma/\mu)$ samples from the distribution.
  \item
  \label{part:Accuracy}
  Produces an estimated mean $\overline{\mu}$ such that
  $
    \left( 1 - \epsilon \right) \mu
    \leq \overline{\mu}
    \leq
    \left( 1 + \epsilon \right) \mu.
  $
\end{itemize}
\end{lemma}
\noindent

An immediate corollary of Lemma~\ref{lem:MeanEstimation} is a routine
$\textsc{ApproxColumnSum}$ for estimating column sums of $A$, where the running
time depends on the column sum itself.  We give the pseudocode for this
estimator in Figure~\ref{fig:EstimateColumnSum}, and then we prove its
correctness and running time in Lemma~\ref{lem:ColumnSumEstimation}.

\begin{figure}[H]
	
	\begin{algbox}
		
		$\textsc{ApproxColumnSum}(A, j, \epsilon, \delta)$
		
		\underline{Input:} matrix $A$ with $r$ rows,
		column index $j$,
    error $\epsilon > 0$, failure probability $\delta > 0$.\\
		Implicit access to the number of rows $r$ and the number of remaining vertices $n$.
		
		\underline{Output:} estimation for  $\textsc{ColumnSum}(A,j)$.
		
		\begin{enumerate}
			
      \item Let $D_{\text{col}}(j)$ denote the distribution for the random variable that:
        \begin{enumerate}
          \item Chooses a row $i \in [r]$ uniformly at random.
          \item Returns the value of $A(i,j)$.
        \end{enumerate}
      \item Set $\sigma \leftarrow 5 \epsilon^{-2} \log\left( 1 / \delta \right)$. 
      \item Return $r \cdot \textsc{EstimateMean}(D_{\text{col}}(j), \sigma)$.
			
		\end{enumerate}
		
	\end{algbox}
	
	\caption{Pseudocode for approximating the column sum of a matrix.}
	
	\label{fig:EstimateColumnSum}
	
\end{figure}

\begin{lemma}
\label{lem:ColumnSumEstimation}
For any $(0,1)$-matrix $A \in \mathbb{R}^{r \times n}$, column index $j$,
error $\epsilon > 0$,
and failure rate~$\delta > 0$,
invoking $\textsc{ApproxColumnSum}$ returns an $\epsilon$-approximation
to $\textsc{ColumnSum}(A,j)$
with probability at least $1 - \delta$
while making
\[
O\left( \frac{r \log\left( 1/\delta \right) }{\textsc{ColumnSum}(A,j) \epsilon^{2}} \right)
\]
oracle calls to the matrix $A$ in expectation.
\end{lemma}

\begin{proof}
Observe that $D_{\text{col}}(j)$ is a Bernoulli distribution
with mean $\textsc{ColumnSum}(A,j)/r$.
The success probability follows directly
from our choice of $\sigma$ and the success probability of
\textsc{EstimateMean} for Bernoulli distributions (Lemma~\ref{lem:MeanEstimation}).
Moreover, the total number of matrix-entry queries is
\[
O\left( {
  \frac{\sigma} {\textsc{ColumnSum}(A,j)/r }} 
\right) 
= 
O \left( {
  \frac{r \log(1/\delta) } {\textsc{ColumnSum}(A,j) {\epsilon^{2}}}
} \right),
\]
as desired.
\end{proof}

We can simulate sampling from the nonzero entries of
the reweighted matrix $B$ by using
the inverse of $\textsc{ApproxColumnSum}$ to estimate $B(i,j)$.
Moreover, by sampling enough entries of $B$ so that their sum is
$O(\log{n}\epsilon^{-2})$, we can accurately estimate
$\textsc{NonzeroColumns}(A)$ by Lemma~\ref{lem:EstimatorCorrectness}.
We give the slow version of this estimator in
Figure~\ref{fig:SlowerEstimateNonZeroColumns}, and we prove its
running time in Lemma~\ref{lem:NonZeroColumnEstimator}.

\begin{figure}[H]
	
	\begin{algbox}
		$\textsc{EstimateNonzeroColumns\_Slow}(A, \epsilon)$
		
		\underline{Input:} oracle access to the matrix $A$
		with $r$ rows and $n$ columns,
		error threshold $\epsilon > 0$.
		
		\underline{Output:} estimation for the number of nonzero columns in $A$.
		
		\begin{enumerate}
			
      \item Compute $\nnz(A)$, the total number of nonzeros in $A$.
      \item Let $D_{\text{global}}$ denote the distribution for the random variable that:
			\begin{enumerate}
        \item Chooses a uniformly random nonzero index $(i,j)$ in $A$
				(by first picking a row with probability proportional to its
				number of nonzeros and then picking a random nonzero entry from that row).
				\item Returns
				\[
				  \frac{1}{\textsc{ApproxColumnSum}(A, j, \epsilon, n^{-10})},
				\]
        where the value of $\textsc{ApproxColumnSum}(A, j, \epsilon, n^{-10})$
				is generated once per each column
				and reused on subsequent calls
				(via storage in a binary search tree).
			\end{enumerate}
    \item Set $\sigma \leftarrow 50\epsilon^{-2} \log(n)$
    \item Return $\nnz(A) \cdot \textsc{EstimateMean}(D_{\text{global}}, \sigma)$.
		\end{enumerate}
		
	\end{algbox}
	
  \caption{Pseudocode for (slowly) estimating the number of nonzero columns of a $(0,1)$-matrix.}
	
	\label{fig:SlowerEstimateNonZeroColumns}
	
\end{figure}

Next, we prove the correctness of $\textsc{EstimateNonzeroColumns\_Slow}$,
and then we bound the expected number of times it generates samples from
$D_{\text{global}}(j)$.

\begin{lemma}
\label{lem:ApproxnonZeroColumnsHelper}
With high probability, the estimation
$\textsc{EstimateNonzeroColumns\_Slow}(A, \epsilon)$
is within a factor of $1 \pm \epsilon$ of the number of nonzero columns of $A$.
\end{lemma}

\begin{proof}
To start, we extract all of the randomness out of
\textsc{EstimateNonzeroColumns\_Slow} by considering running all calls to
$\textsc{ApproxColumnSum}(A, j, \epsilon, n^{-10})$ beforehand.
By Lemma~\ref{lem:ColumnSumEstimation}, 
with high probability, for each column $j$ we have
\[
  \left( 1 - \epsilon \right)
  \textsc{ColumnSum}\left(A, j \right)
  \leq
  \textsc{ApproxColumnSum}\left(A, j, \epsilon, n^{-10}\right)
  \leq
  \left( 1 + \epsilon \right)
  \textsc{ColumnSum}\left(A, j \right).
\]
Therefore, by Lemma~\ref{lem:EstimatorCorrectness}
it follows that
\[
  \left(1 - 2\epsilon\right) \frac{\textsc{NonzeroColumns}(A)}{\nnz(A)}
  \le \mu\left(D_{\text{global}}\right)
  \le \left(1 + 2\epsilon\right) \frac{\textsc{NonzeroColumns}(A)}{\nnz(A)},
\]
for sufficiently small $\epsilon$.
Incorporating the accuracy guarantee from Lemma~\ref{lem:MeanEstimation} gives
\[
  \left(1 - 4\epsilon\right) \frac{\textsc{NonzeroColumns}(A)}{\nnz(A)}
  \leq
  \textsc{EstimateMean}\left(
    D_{\text{global}}, \sigma \right)
  \leq
  \left(1 + 4\epsilon\right) \frac{\textsc{NonzeroColumns}(A)}{\nnz(A)}.
\]
The desired bound follows by quartering $\epsilon$ and multiplying by $\nnz(A)$.
\end{proof}

\begin{proof}[Initial Proof of Lemma~\ref{lem:NonZeroColumnEstimator}
(using $O(r \log^{2}{n} \epsilon^{-4})$ operations).]

The correctness is a consequence of Lemma~\ref{lem:ApproxnonZeroColumnsHelper},
so we must bound the total number of queries to entries of $A$.
Using Lemma~\ref{lem:MeanEstimation} and Lemma~\ref{lem:EstimatorCorrectness},
the expected number of queries made to
$D_{\text{global}}$ is
\[
  O\left(
  \frac{\nnz(A) \log{n}}
  {\textsc{NonzeroColumns}(A) \epsilon^{2}}
  \right).
\]
Therefore, it suffices to bound the expected cost of each sample generated from 
$D_{\text{global}}$.

Applying Lemma~\ref{lem:ColumnSumEstimation} to each column $j$,
the expected number of queries to $A$ made by $\textsc{ColumnSum}(A,j)$ is
\[
O\left( \frac{r \log{n}}{\textsc{ColumnSum}\left(A,j\right) \epsilon^{2}} \right).
\]
Summing over all $\textsc{ColumnSum}(A,j)$ nonzero entries in column $j$
gives
\[
  O\left( \frac{r \log{n}}{\epsilon^{2}} \right)
\]
queries to $A$ per nonzero column.
It follows that the expected number of queries to $A$
per sample generated from $D_{\text{global}}$ is
\[
  O\left(\frac{\textsc{NonzeroColumns(A)}}{\nnz(A)} \cdot
   \frac{r \log{n}}{\epsilon^{2}} \right).
\]
Multiplying this by the expected number of queries to $D_{\text{global}}$
gives the overall result.
\end{proof}

We note that this sample complexity bound also holds with
high probability (instead of only in expectation) by invoking Chernoff bounds.
This is because the cost of each query to $D_{\text{global}}$ is bounded
by $O(r \log{n} \epsilon^{-2})$ and the overall cost bound is larger
by a factor of at least $\Omega(\log{n})$.

\subsection{Estimating the Mean of a Distribution}
\label{subsec:MeanEstimation}

We now provide the details for the mean estimation algorithm, which proves 
the correctness of the column sum estimator.
We analyze the following scheme:
\begin{enumerate}
  \item Generate an infinite stream of i.i.d.\ samples $X_1, X_2,\dots$ from
    any distribution $D$ over $[0,1]$.
  \item Let $\variable{counter} = \min
    \left\{t \ge 0 : {\sum_{i=1}^t X_i \ge \sigma } \right\}$.
	\item Output $\sigma/\variable{counter}$.
\end{enumerate}
This process generates more samples than \textsc{EstimateMean}
in Figure~\ref{fig:EstimateMean}, but
the extra evaluations happen after the subroutine terminates
and thus does not affect the outcome.

Let $\mu$ be the (hidden) mean of the distribution $D$.
For any error $\epsilon > 0$, define the two cutoffs
\begin{equation*}
\label{eq:L}
  L\left(D, \epsilon\right)
  \defeq
  \frac{\sigma }{{(1 + \varepsilon ) \mu }}
\end{equation*}
and
\begin{equation*}
\label{eq:R}
  R\left(D, \epsilon\right)
  \defeq
  \frac{\sigma }{{(1 - \varepsilon ) \mu }}.
\end{equation*}
\noindent
For convenience we write $L=L(D,\epsilon)$ and $R=R(D,\epsilon)$.
We claim that if $L~\le~\variable{counter}~\le R$,
then the estimation we output will be sufficiently accurate.  Therefore, we
first bound the probabilities of the complementary events $\variable{counter} <
L$ and $\variable{counter} > R$.

\begin{lemma}
\label{lem:ProbCounterTooSmall}
Let $D$ be any distribution over $[0,1]$.
For any sequence $X_1, X_2, X_3, \dots$
of i.i.d.\ random variables generated from $D$
and any choice of $\sigma \ge 0$, we have
\[
\Pr\left[
  {\sum_{i=1}^{L} X_i  \ge \sigma } \right]
\leq
\exp \left( - \frac{\epsilon^2 \sigma}{4} \right).
\]
\end{lemma}

\begin{proof}
By the linearity of expectation we have
\[
  \E\left[
    \sum_{i=1}^{L} X_i \right]
  = L \mu.
\]
Since $X_{1},X_{2},\dots,X_{L}$ are independent random variables
with values in $[0,1]$, Chernoff bounds give
\begin{align*}
\Pr\left[
  \sum_{i=1}^{L} X_i
  \geq
  \left(1 + \epsilon \right) L\mu 
\right]
  &\leq
  \exp \left( - \frac{\epsilon^2 L\mu}{3} \right) \\
  &\leq \exp \left( - \frac{\epsilon^2 \sigma}{4} \right)
\end{align*}
by letting $\sigma=(1+\epsilon)L\mu$ and considering
$\epsilon$ sufficiently small.
This completes the proof.
\end{proof}

\begin{lemma}
\label{lem:ProbCounterTooBig}
Let $D$ be any distribution over $[0,1]$.
For any sequence $X_1, X_2, X_3, \dots$
of i.i.d.\ random variables generated from $D$
and any choice of $\sigma \ge 0$, we have
\[
  \Pr\left[
    \sum_{i=1}^R X_i  \leq \sigma  \right]
  \leq
  \exp \left( - \frac{\epsilon^2 \sigma}{4} \right).
\]
\end{lemma}

\begin{proof}
Consider the proof of Lemma~\ref{lem:ProbCounterTooSmall}, and use a Chernoff
bound for the lower tail instead.
\end{proof}

\begin{proof}[Proof of Lemma~\ref{lem:MeanEstimation}]
We will show that the estimator behaves as intended when
$L \le \variable{counter} \le R$.
Considering the complementary events, it is easy to see that
\[
\Pr\left[\variable{counter} \le L \right] \leq \Pr\left[{\sum_{i=1}^{L} X_i  \ge \sigma } \right]
\]
and
\[
\Pr\left[\variable{counter} \geq R \right] \leq \Pr\left[{\sum_{i=1}^{R} X_i  \le \sigma } \right].
\]
Therefore, by Lemma~\ref{lem:ProbCounterTooSmall} and Lemma~\ref{lem:ProbCounterTooBig}
it follows that
\begin{align*}
  \Pr\left[L \le \variable{counter} \le R \right] &\ge
  1 - 2\exp\left(-\frac{\epsilon^2 \sigma}{4}\right) \\
  &\ge 1 - \exp\left(-\frac{\epsilon^2\sigma}{5}\right).
\end{align*}

Assume that $L \le \variable{counter} \le R$ and recall the definitions
of $L$ and $R$.
It follows that the number of samples generated is $O(\sigma/\mu)$.
To prove that $\sigma/\variable{counter}$ is an accurate estimate,
observe that
\[
  \frac{\sigma}{(1+\epsilon)\mu} \le \variable{counter} \le \frac{\sigma}{(1-\epsilon)\mu},
\]
and therefore
\[
  \left(1-\epsilon\right)\mu \le \frac{\sigma}{\variable{counter}} \le \left(1+\epsilon\right)\mu.
\]
This completes the proof.
\end{proof}

\subsection{Improving the Error Bounds Through a More Holistic Analysis}
\label{subsec:DegreeEstimation_Better}

We now give a better running time bound by combining the analyses of the two
previous estimators in a more global setting.  Pseudocode for this final
estimation routine is given in Figure~\ref{fig:EstimateNonZeroColumns}.

\begin{figure}[H]
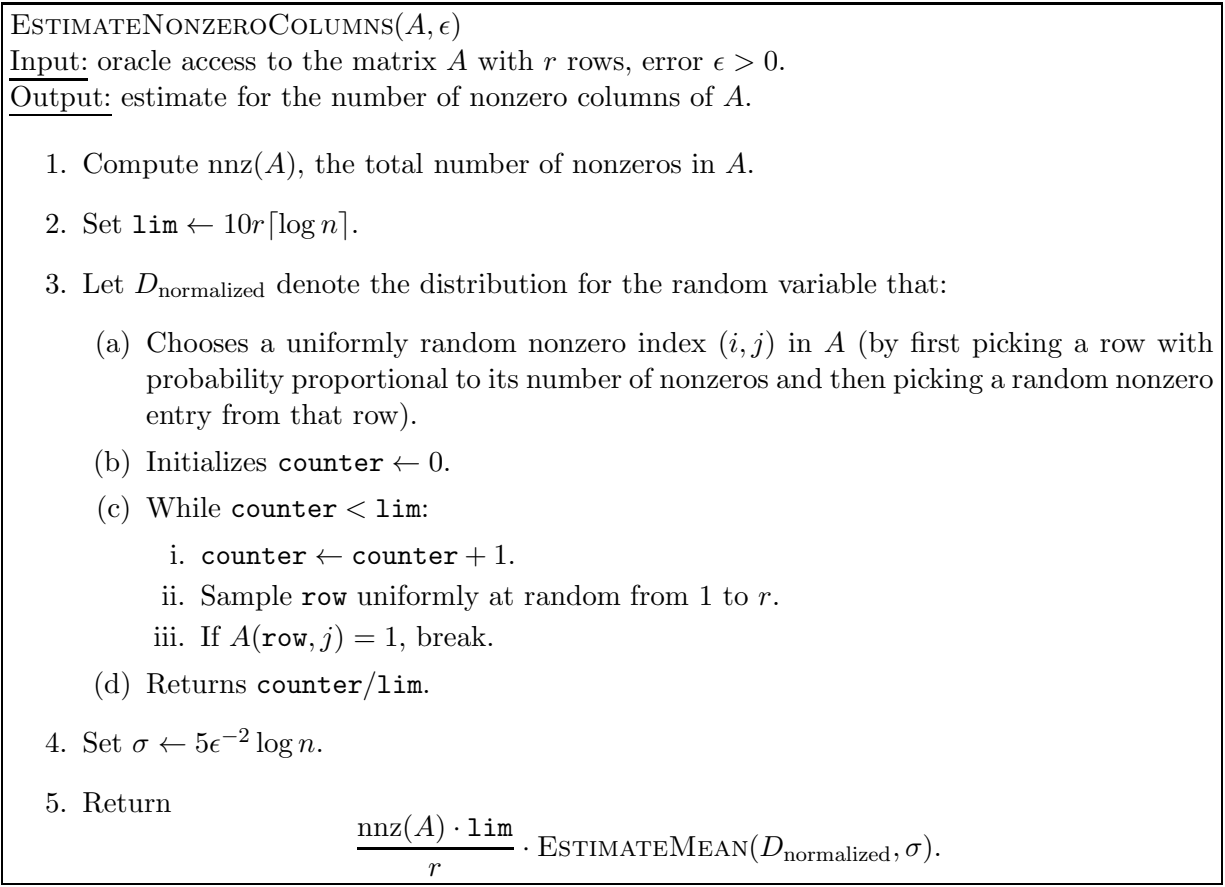


\begin{algbox}
$\textsc{EstimateNonzeroColumns}(A, \epsilon)$

\underline{Input:} oracle access to the matrix $A$
with $r$ rows, 
error $\epsilon > 0$.

\underline{Output:} estimate for the number of nonzero columns of $A$.

\begin{enumerate}

\item Compute $\nnz(A)$, the total number of nonzeros in $A$.
\item Set $\variable{lim} \leftarrow 10r \lceil\log{n}\rceil$.
\item Let $D_{\text{normalized}}$ denote the distribution for the random variable
  that:
  \begin{enumerate}
    \item Chooses a uniformly random nonzero index $(i,j)$ in $A$
      (by first picking a row with probability proportional to its number
      of nonzeros and then picking a random nonzero entry from that row).
    \item Initializes $\variable{counter} \leftarrow 0$.
    \item While $\variable{counter} < \variable{lim}$:
       \begin{enumerate}
         \item $\variable{counter} \leftarrow \variable{counter} + 1$.
         \item Sample $\variable{row}$ uniformly at random from $1$ to $r$.
         \item If $A(\variable{row},j) = 1$, break.
    \end{enumerate}
    \item Returns $\variable{counter}/\variable{lim}$.
\end{enumerate}
\item Set $\sigma \leftarrow 5\epsilon^{-2}\log{n}$.
\item Return
  \[
    \frac{\nnz(A) \cdot \variable{lim}}{r}\cdot\textsc{EstimateMean}(D_{\text{normalized}}, \sigma).
  \]
\end{enumerate}

\end{algbox}

\caption{Pseudocode for a faster estimation of the number of nonzero columns of a matrix.}
\label{fig:EstimateNonZeroColumns}
\end{figure}

\noindent
The essence of this algorithm can be better understood by analyzing a simpler
version of $D_{\text{normalized}}$.
The difference between these two distributions
is that we artificially force $D_{\text{normalized}}$ to be a distribution
over $[0,1]$ (first by truncating and then normalizing)
so that we can use the $\textsc{EstimateMean}$ algorithm.
We show that our threshold of $O(r \log{n})$ for truncating the number of
sampled rows can perturb the expected value by at most a factor of $1/\text{poly}(n)$.

\begin{definition}
\label{def:DCombined}
We define the simpler variant $D_{\text{simple}}$ of $D_{\text{normalized}}$
as follows:
\begin{enumerate}
  \item Sample a nonzero index $(i, j)$ from $A$ uniformly at random.
  \item Return the minimum of $10r\lceil\log{n}\rceil$ and
    the number of random rows $\variable{row}$ until $A(\variable{row}, j)=1$.
\end{enumerate}
\end{definition}

Now we analyze the expected value of $D_{\text{simple}}$ and relate it to
that of $D_{\text{normalized}}$ to
prove the correctness of \textsc{EstimateNonzeroColumns} and
bound the time needed to sample from $D_{\text{normalized}}$.

\begin{lemma}
\label{lem:DNormalizedExpectation}
If $X$ is a random variable drawn from $D_{\text{normalized}}$, then
\[
\left( 1 - \frac{1}{n} \right)
  \frac{r \cdot \textsc{NonzeroColumns}(A)}{\variable{lim} \cdot \nnz(A)}
\leq
  \E\left[X\right]
\leq
  \frac{r \cdot \textsc{NonzeroColumns}(A)}{\variable{lim} \cdot \nnz(A)},
\]
and the expected cost of each sample generated from $D_{\variable{normalized}}$ is
\[
  O\left( \frac{r \cdot \textsc{NonzeroColumns}(A)}{\nnz(A)} \right).
\]
\end{lemma}

\begin{proof}
For each column $j$, denote by $\nnz(A[:,j])$ the number of nonzero entries
in the column, and
let the probability of picking a nonzero entry from this column be
\[
  p_j \defeq \frac{\nnz(A[:,j])}{r}.
\]
When considering columns selected by $D_{\text{combined}}$ and
$D_{\text{simple}}$, the column necessarily has a nonzero entry so $p_j > 0$.

Next, define the random variable $H_j$ to be the number of times it takes to
independently sample a row $\variable{row}$ between $1$
and $r$ uniformly at random such that $A(\variable{row}, j) = 1$.
For all integers $k \ge 1$,
\[
  \Pr\left[H_j = k\right] = \left(1 - p_j\right)^{k-1} p_j.
\]
It follows that the expected value of $H_j$ is
\begin{align*}
  \E\left[H_j\right] &= \sum_{k=1}^\infty k \Pr\left[H_j = k\right]\\
    &= p_j \sum_{k=1}^\infty k \left(1 - p_j\right)^{k-1} \\
    &= \frac{p_j}{\left(1 - \left(1-p_j\right)\right)^2} \\
    &= \frac{r}{\nnz\left(A[:,j]\right)},
\end{align*}
where the second to last line uses the Maclaurin series
\[
  \sum_{k=1}^\infty k x^{k-1} = \frac{1}{(1-x)^2},
\]
for all $|x| < 1$. Note that we can apply this formula because $p_j > 0$.

To account for truncation, consider the random variable
$Z_j = \min(H_j, \variable{lim})$ and observe that
\begin{align*}
  \E\left[Z_j\right] &= \sum_{k=1}^{\variable{lim}} k \Pr\left[H_j = k\right]
                     + \sum_{k=1}^{\infty} (\variable{lim}+k - k) \Pr\left[H_j = \variable{lim}+k\right]\\
     &= \E\left[H_j\right] - \left(1-p_j\right)^{\variable{lim}}p_j \sum_{k=1}^\infty k\left(1-p_j\right)^{k-1}\\
     &= \E\left[H_j\right] - \frac{\left(1-p_j\right)^{\variable{lim}}}{p_j},
\end{align*}where we again use the Maclaurin series above.
Recalling that $\nnz(A[:,j]) \ge 1$ and $r \le n$, it follows from the
definition of $p_j$ that
\begin{align*}
  \frac{\left(1-p_j\right)^{\variable{lim}}}{p_j} &\le n \left(1-p_j\right)^{\variable{lim}}\\
  &\le n \exp\left(-\frac{\nnz(A[:,j])}{r} \cdot 10 r \log{n}\right)\\
  &= n \left(\frac{1}{n}\right)^{10 \cdot \nnz(A[:,j])}\\
  &\le \frac{1}{n^9}.
\end{align*}
Therefore, by truncating the number of row samples, we deviate from $\E[H_j]$ by at
most $1/\text{poly}(n)$. Letting $Z$ be a random variable drawn from
$D_{\text{simple}}$, we have the conditional expectation bound
\begin{align*}
  \frac{r}{\nnz\left(A[:,j]\right)} - \frac{1}{n^9}
  \le \E\left[Z \mid \text{$Z$ chooses column $j$}\right]
  \le \frac{r}{\nnz\left(A[:,j]\right)}.
\end{align*}
If we now consider the probability that $Z$ selects column $j$, it follows that
\begin{align*}
  \frac{r \cdot \textsc{NonzeroColumns}(A)}{\nnz(A)} \left(1 - \frac{1}{n^7}\right)
  \le \E\left[Z\right]
  \le \frac{r \cdot \textsc{NonzeroColumns}(A)}{\nnz(A)},
\end{align*}
since $\nnz(A) = O(n^2)$.
To prove the claim for the second distribution $D_{\text{normalized}}$, we
use the same argument and the additional fact that $\variable{lim} = O(n^2)$.
The expected running time per sample follows from the expected value of
$D_{\text{simple}}$.
\end{proof}

\begin{proof}[Proof of Lemma~\ref{lem:NonZeroColumnEstimator}]
Since $D_{\variable{normalized}}$ is a distribution over $[0,1]$, we can use
\textsc{EstimateMean} to approximate the expected value within a factor of
$\epsilon$ with high probability by Lemma~\ref{lem:MeanEstimation}.
The correctness of \textsc{EstimateNonzeroColumns} follows from our choice
of $\sigma$ and Lemma~\ref{lem:DNormalizedExpectation}.

It remains to bound the running time of the estimator.
By Lemma~\ref{lem:MeanEstimation} and Lemma~\ref{lem:DNormalizedExpectation},
the expected number of times we sample $D_{\text{normalized}}$ is
\[
  O\left(\frac{\nnz(A) \cdot \variable{lim} \cdot \log{n} \epsilon^{-2}}{r \cdot \textsc{NonzeroColumns}(A)}\right).
\]
The expected cost per sample generated from $D_{\text{normalized}}$ is
\[
  O\left(\frac{r\cdot \textsc{NonzeroColumns}(A)}{\nnz(A)}\right)
\]
by Lemma~\ref{lem:DNormalizedExpectation}.
Multiplying these expectations gives that the total expected running time is
\[
O\left( \variable{lim} \cdot \log{n} \epsilon^{-2} \right)
= O\left(r \log^2{n} \epsilon^{-2} \right).
\]
Furthermore, since the cost per sample from $D_{\text{normalized}}$
is bounded by $\lim = O(r \log{n})$,
it follows that the running time is concentrated around this value
with high probability.
\end{proof}

%% file: DynamicGraphs.tex
\section{Maintaining Graphs Under Pivots}
\label{sec:DynamicGraphs}

In this section we first show how to efficiently maintain the component graph
under pivots in such a way that supports component and remaining neighborhood queries.
This proves Lemma~\ref{lem:DegreeEstimationDS}, which in turns completes the proof
of Theorem~\ref{thm:DegreeEstimation}.
Then we spend the rest of the section demonstrating how to maintain a dynamic
1-neighborhood sketch (described in Definition~\ref{def:Sketch}) of a fill
graph as it undergoes vertex eliminations to prove Theorem~\ref{thm:DataStructureMain}.

For convenience, we restate our claim about supporting dynamic component graphs
$\Gcomp{}$.

\DegreeEstimationDS*

\begin{proof}
We maintain the adjacency list of $\Gcomp{}$ explicitly, where each node
stores its state as well as its neighbors in balanced binary search trees.
When we pivot a remaining vertex $v$, we examine all all of its neighbors
that are component vertices (i.e.\ $\Ncomp{}(u)$) and merge the neighborhood
lists of these vertices.
By always inserting elements from the smaller list
into the larger, we can guarantee that each element
is inserted at most $O(\log{n})$ times.
The total number of elements in the
neighborhood lists is $m$, so 
since each insertion costs $O(\log n)$ time, 
it follows that the total cost across all $m$ edges is $O(m \log^2{n})$.

When a vertex is pivoted, we also need to move it
from the remaining vertex list to the component vertex list for each of its
neighbors.  This can be done by iterating over all the edges of the vertex
once. The cost is $O(m)$ because each vertex is pivoted at most once and
prior to pivoting no edges are added to the graph.
By maintaining all lists using balanced binary search trees,
we can make all updates and sample a random remaining (or component)
neighbor in $O(\log{n})$ time.
A global list that tracks all remaining and component
vertices also allows for $O(\log{n})$ time uniform sampling.
\end{proof}

We now focus on proving Theorem~\ref{thm:DataStructureMain}.
We maintain a 1-neighborhood $\ell_0$-sketch data structure for a fill graph
as it undergoes pivots (starting with the original graph) similarly to how
we maintain the adjacency list of the component graph.
Because the minimum key value $R(v)$ in the 1-neighborhood of a vertex $u$
continually changes, we show how to track the minimizer of a vertex via an
eager-propagation routine. This protocol informs the neighbors of a pivoted
vertex about its minimum key, which ultimately propagates minimum key values
throughout the graph as needed.

In Figure~\ref{fig:DynamicGlobalVar}, we give a brief description about the
data structures we use to achieve this.
When we refer to maintaining sets of vertices, the underlying data structure
is a balanced binary search tree.
Recall that min heaps (e.g.\ binary heaps) of size $O(n)$
support the methods $\textsc{Min}$, $\textsc{Insert}$, and $\textsc{Delete}$
which require $O(1)$, $O(\log n)$, and $O(\log n)$ time respectively.
Additionally, we use a standard subroutine $\textsc{HeapMerge}$ 
to merge two heaps in $O(\log^2 n)$ time.

\begin{figure}[H]
  \begin{algbox}
    \begin{enumerate}
      \item A set $\Vrem{}$ containing the remaining vertices.
      \item A set $\Vcomp{}$ containing the component vertices.
      
      \item For each $x \in \Vcomp{}$, a corresponding min heap
      \[
        \variable{remaining}[x]
      \]
      that contains the key values $R(v)$ of its remaining neighbors $v \in \Nrem{}(x)$.

  	  \item For each $u \in \Vrem{}$, a corresponding min heap
      \[
         \variable{fill}[u]
      \]
      that contains the union of $\variable{remaining}[x].\textsc{Min}()$
      for each component vertex $x \in \Ncomp{}(u)$,
      as well as the key values of the vertices in $\Nrem{}(u)$.
    \end{enumerate}
    
  \end{algbox}
  
  \caption{Data structures needed to maintain $\Gcomp{}$ and an $\ell_0$-sketch of $\Gfill{}$ under vertex pivots.}
  \label{fig:DynamicGlobalVar}
\end{figure}

As a brief intuition behind the algorithm,
consider the case where no vertex is deleted,
but we merge neighborhoods of vertices.
In this case, as the neighborhood of a particular vertex grows,
the expected number of times the minimum $R$ value in
this neighborhood changes is $O(\log n)$.
To see this, consider the worst case where each time,
the neighborhood of $v$ increases by 1.
Then the expected number of changes in the minimum is
\[
1/2 + 1/3 + \ldots +1/(n-1) \leq O(\log n).
\]
The major difficulty dealing with this is that
deletions reduce degrees.
In particular, it is possible for the min at some vertex
to change $\Omega(n)$ times due to repeated deletions.
As a result, we can only bound the total, or average
number of propagations.
This leads to a much more involved amortized
analysis, where we also use backwards analysis to explicitly
bound the probability of each informing operation.

Given a component graph $\Gcomp{t}$ and a (remaining)
vertex $u$ to be pivoted, we use the routine \textsc{PivotVertex}
to produce a new graph $\Gcomp{t+1}$.
In terms of the structure of the graph, our routine does the same
thing as the traditional quotient graph model for symmetric
factorization~\cite{GeorgeL81}.

Therefore we turn our attention to the problem of maintaining the minimum $R$
values of the neighborhoods.
For a subset of vertices $V' \subseteq \Vcomp{}$,
let $\Rmin(V')$ denote the minimum $R$ value among all its vertices.
Specifically, we want to maintain the values
$\Rmin(\Nrem{}(w))$ for every $w \in \Vcomp{}$ and
$\Rmin(\Nfill{}(v))$ for every $v \in \Vrem{}$.
This update procedure is basically a notification mechanism.
When the status of a vertex changes, we update the data structures of its
neighbors correspondingly.
The \variable{fill}[u] heap will then give $\Rmin(\Nfill{}(u))$ and be
used to estimate the fill-degree of each remaining vertex
as described in Section~\ref{sec:Sketching}.

Suppose a remaining vertex $v$ is pivoted.
Then, for a component vertex $w$, the content of $\variable{remaining}[w]$
changes only if $v$ is its neighbor.
Pseudocode of this update (\textsc{PivotVertex}) is given in Figure~\ref{fig:PivotVertex}.
In particular, since $v$ is no longer a remaining vertex, its entry needs to be removed from
$\variable{remaining}[z]$.
Since $v$ is now a component vertex,
we need to construct $\variable{remaining}[v]$,
and update the $\variable{fill}$ heaps of its remaining neighbors appropriately.
Furthermore, if $R(v)$ was the minimum element in $\variable{remaining}[w]$,
this is no longer the case and the other remaining neighbors of $w$ need to be notified of
this (so they can update their $\mathit{fill}$ heaps).
This is done via the call to $\textsc{InformRemaining}$ in
Line~\ref{algline:pivot-inform-remaining} of the algorithm.
The last step consists of melding the (now component) vertex $v$ with its existing component
neighbors via calls to $\textsc{Meld}$.
The pseudocode for this routine is in Figure~\ref{fig:Meld}.
Note that, at all times, we make a note of any remaining vertex
whose $\minimizer$ is updated due to the pivoting.
\begin{figure}[H]

\begin{algbox}
$\textsc{PivotVertex}(v)$

  \underline{Input}: (implicitly as a global variable)
  a component graph $\Gcomp{t} = \langle {\Vrem{}},{\Vcomp{}},\Ecomp{}\rangle$
  along with associated data structures.\\
  A vertex $v \in {\Vrem{}}$ to be pivoted, 

  \underline{Output}: A list of vertices in $\Vrem{}$ whose $\minimizer$s have changed.

\begin{enumerate}
\item Initialize $\variable{changed\_list} \leftarrow \emptyset$.
\item Create an empty min-heap $\variable{remaining}[v]$
\item For each vertex $y \in \Nrem{}(v)$ in lexicographical order
\begin{enumerate}
\item  $\variable{fill}[y].\textsc{Delete}(R(v))$ \label{algline:remove-from-remaining}
\item $\variable{remaining}[v].\textsc{Insert}(R(y))$
\item If $R(v)$ was the old minimum in $\variable{fill}[y]$:

      $\variable{changed\_list} \leftarrow \variable{changed\_list} ~\cup~ \{y\}$
\end{enumerate}
\item For each vertex $y \in \Nrem{}(v)$ in lexicographical order
\begin{enumerate}
	\item  $\variable{fill}[y].\textsc{Insert}(\variable{remaining}[v].\textsc{Min}())$ (if not already present)
	\item If the minimum in $\variable{fill}[y]$ changes:
	
	$\variable{changed\_list} \leftarrow \variable{changed\_list} ~\cup~ \{y\}$
\end{enumerate}
\item For each vertex $w \in \Ncomp{}(v)$ in lexicographical order
\begin{enumerate}
	\item  $\variable{remaining}[w].\textsc{Delete}(R(v))$.
	\label{algline:remove-from-component}
	\item If $R(v)$ was the old minimum in $\variable{remaining}[w]$:
	
	$\variable{changed\_list} \leftarrow \variable{changed\_list}
	~\cup~ \textsc{InformRemaining}(w,R(v),\variable{remaining}[w].\textsc{Min}())$\label{algline:pivot-inform-remaining}
	\item  $\textsc{Meld}(v, w)$;
\end{enumerate}
\item Update $\Vcomp{}$,$\Vrem{}$ and $\Ecomp{}$ to form $\Gcomp{t+1}$;
\item Return $\variable{changed\_list}$.
\end{enumerate}
\end{algbox}

\caption{Pseudocode for pivoting a vertex}
\label{fig:PivotVertex}

\end{figure}

\begin{figure}[H]

\begin{algbox}
$\textsc{Meld}(v, w)$

\underline{Input}: (implicitly as a global variable)
A graph state $G = \langle {\Vrem{}},{\Vcomp{}},E\rangle$
along with associated data structures.\\
Two component vertices $v$ (the pivoted vertex) and $w$ to be melded. 

\underline{Output}: A list of vertices in $\Nrem{}(v) \cup \Nrem{}(w)$ whose $\minimizer$s have changed.

\begin{enumerate}
	\item Initialize $\variable{changed\_list} \leftarrow \emptyset$.
	\item If
	$\variable{remaining}[v].\textsc{Min}() < \variable{remaining}[w].\textsc{Min}()$
	\begin{enumerate}
		\item $\variable{changed\_list} \leftarrow \textsc{InformRemaining}(w,\variable{remaining}[w].\textsc{Min}(),\variable{remaining}[v].\textsc{Min}())$;\label{algline:informvaboutu}
	\end{enumerate}
	\item Else If
	$\variable{remaining}[w].\textsc{Min}() < \variable{remaining}[v].\textsc{Min}()$
	\begin{enumerate}
		\item $\variable{changed\_list} \leftarrow \textsc{InformRemaining}(v,\variable{remaining}[v].\textsc{Min}(),\variable{remaining}[w].\textsc{Min}())$;\label{algline:informuaboutv}
	\end{enumerate}
	\item $\variable{remaining}[v] \leftarrow \textsc{HeapMerge}(\variable{remaining}[v],\variable{remaining}[w])$
	\item Return $\variable{changed\_list}$.
\end{enumerate}

\end{algbox}

\caption{Pseudocode for melding two component vertices,
	and informing their neighbors of any changes in the minimizers
	of $\Nrem{}$.}
\label{fig:Meld}

\end{figure}

For every component vertex $w$ such that $R(v)$ is the minimum value in $\variable{remaining}(w)$,
the routine \textsc{InformRemaining} (Pseudocode in Figure~\ref{fig:InformRemaining}) is
responsible for updating the contents in the \variable{fill} heaps of remaining
vertices adjacent to $w$.
This routine is also required when we merge two component vertices
in the algorithm \textsc{Meld}, since there are now more entries
in the \variable{fill} heaps of adjacent remaining vertices.

\begin{figure}[H]
	
	\begin{algbox}
		$\textsc{InformRemaining}(w, R_{old}, R_{new})$
		
		\underline{Input}: (implicitly as a global variable)
		a component graph $\Gcomp{} = \langle {\Vrem{}},{\Vcomp{}},\Ecomp{}\rangle$
		along with associated data structures;\\
		a vertex $w \in \Vcomp{}$ that's causing updates;\\
		old and new values for $\Rmin(\Nrem{}(w))$: $R_{old}$ and $R_{new}$.
		
		\underline{Output}: A list of vertices $v \in \Nrem{}(w)$ whose $\minimizer$s have changed.
		
		\begin{enumerate}
			\item Initialize $\variable{changed\_list} \leftarrow \emptyset$.
			\item For each $v \in \Nrem{}(w)$
			\begin{enumerate}
				\item Delete the entry $R_{old}$ from $\variable{fill}[v]$
				if it exists
				\item Add the entry $R_{new}$ to $\variable{fill}[v]$
				\item If $\variable{fill}[v].\textsc{Min}()$ changed,
				$\variable{changed\_list} \leftarrow \variable{changed\_list} \cup \{ v \}$.
			\end{enumerate}
			\item Return $\variable{changed\_list}$.
		\end{enumerate}
		
	\end{algbox}
	
	\caption{Pseudocode for propagating to remaining vertex neighbors}
	\label{fig:InformRemaining}
\end{figure}

We break down the cost of calls to \textsc{InformRemaining} into two parts: when it is invoked by
\textsc{PivotVertex}, and when it is invoked by \textsc{Meld}.
The first type of calls happens only when a remaining vertex $v$ is pivoted, and $v$
is the minimum entry of the \textit{remaining} heap of a component vertex.
The following lemma gives an upper bound on the expected cost of such calls to
\textsc{InformRemaining} by arguing that this event happens with low
probability.

\begin{lemma}
\label{lem:RemainingUpdatesDueToDeletion}
  The expected total number of updates to remaining vertices made by
  \textsc{InformRemaining} when invoked from \textsc{PivotVertex}
  (Line~\ref{algline:pivot-inform-remaining})
  over any sequence of $n$ pivots that are independent of the $R$ values is $O(m)$.
\end{lemma}

\begin{proof}
  Let $\Gcomp{}$ be the component graph at a certain instant
  in the algorithm.
  Let $v \in \Vrem{}$ be the vertex to
  be pivoted, and let $w \in \Ncomp{}(v)$ be a neighboring component vertex.
  We only invoke \textsc{InformRemaining} if $R(v)$ is the minimum value in
  $\variable{remaining}[w]$, which occurs with probability
  $1/|\Nrem{}(w)|$ and would cost
  $O(|\Nrem{}(w)|)$ updates.
  Therefore the expected number of updates is only $O(1)$ for each
  edge between a remaining vertex and a component vertex.
  When a remaining vertex $v$ is pivoted, its degree is the same as in the
  original graph.
  Therefore the number of edges between $v$ and a
  component vertex is bounded by the degree of $v$
  and hence the total expected number of updates is $O \left( \sum_{v \in V} \deg(v) \right) = O(m)$.
\end{proof}

The calls to $\textsc{Meld}$ are the primary bottlenecks in the running time,
but will be handled similarly.
Its pseudocode is given in Figure~\ref{fig:Meld}.

We will show that the expected number of vertices updated by
$\textsc{InformRemaining}$ that result from any fixed sequence of calls to
$\textsc{Meld}$ is bounded by $O(m \log{n})$.
We first analyze the number of updates during a single meld in the following
lemma.

\begin{lemma}
  \label{lem:remaining-updates}
  Let $u$ and $v$ be two component vertices in a graph stage $\Gcomp{}$.
  Then the expected number of updates to vertices by
  $\textsc{InformRemaining}$ when melding $u$ and $v$ is at most:
  \begin{align*}
  \frac{2\left|\Nrem{}(u)\right|
    \cdot \left|\Nrem{}(v)\right|}
  {\left|\Nrem{}(u)\right|
    +
    \left|\Nrem{}(v)\right|
  },
  \end{align*}
\end{lemma}

\begin{proof}
  Let us define:
  \begin{align*}
    n_{common} & = \left|\Nrem{}(u) \cap \Nrem{}(v)\right|,\\
    n_{u} & = \left|\Nrem{}(u) \setminus \Nrem{}(u)\right|,\\
    n_{v} & = \left|\Nrem{}(u) \setminus \Nrem{}(v)\right|.
  \end{align*} 
  If the minimum $R$ value is generated by a vertex from
  $\Nrem{}(u) \cap \Nrem{}(v)$, then no cost is incurred.
  If it is generated by a vertex from
  $\Nrem{}(u) \setminus \Nrem{}(v)$, we need to update the
  every vertex in $\Nrem{}(v)$
  (line~\ref{algline:informvaboutu}).
  This happens with probability
  \begin{align*}
    \frac{n_{u}}{n_{common} + n_{u} + n_{v}}
    &\le 
    \frac{n_{u} + n_{common}}{2n_{common} + n_{u} + n_{v}}
    \\&=\ 
    \frac{
      \left|\Nrem{}(u)\right|
    }{
      \left|\Nrem{}(u)\right| + \left|\Nrem{}(v)\right|
    }.
  \end{align*}
  Therefore the expected number of updates is bounded by:
  \begin{align*}
    \frac{\left|\Nrem{}(u)\right|
      \cdot \left|\Nrem{}(v)\right|}
    {\left|\Nrem{}(u)\right|
     +
     \left|\Nrem{}(v)\right|
    },
  \end{align*}
  and we get the other term (for updating $u$'s neighborhood) similarly.
\end{proof}

This allows us to carry out an amortized analysis for the number of updates to remaining
vertices.
We will define the potential function of an intermediate component graph during elimination
in terms of the degrees of component vertices \emph{in the original graph $G$}, in
which adjacent component vertices are \emph{not} contracted.
Let $u^{\circ}$ denote the set of vertices in $V(G)$ which have
been melded into $u$ in $\Gcomp{}$.
\begin{align*}
  \Phi(\Gcomp{})
  \defeq
  \sum_{u \in \Vcomp{}} D(u) \log \left(D(u)\right) ,
\end{align*}
where $D(u)$ for a vertex $u\in \Vcomp{}$ is defined to be
\begin{align*}
  D(u)=\sum_{v \in u^{\circ}} \deg_{G}(u).
\end{align*}
This function starts out at $0$, and can be at most $m\log{n}$.

\begin{lemma}
  \label{lem:potential-increase}
  When melding two neighboring component vertices in a graph $\Gcomp{t}$ to create
  $\Gcomp{t+1}$,
  the expected number of vertex updates by $\textsc{InformRemaining}$
  is at most
  \begin{align*}
    2\left(\Phi (\Gcomp{t+1})-\Phi (\Gcomp{t} )\right).
  \end{align*}
\end{lemma}

\begin{proof}
When melding two component vertices $u$ and $v$ in $\Gcomp{t}$ to form $\Gcomp{t+1}$, the change in
potential is given by
\begin{align*}
\Phi(\Gcomp{t+1})-\Phi(\Gcomp{t}) = (D(u)+D(v))\log(D(u)+D(v))
-D(u)\log D(u)
-D(v)\log D(v).
\end{align*}
On the other hand, by Lemma~\ref{lem:remaining-updates} the expected number of
remaining vertices updated is
\begin{align*}
\frac{2\left|\Nrem{}(u)\right|
\cdot \left|\Nrem{}(v)\right|}
{\left|\Nrem{}(u)\right|
+
\left|\Nrem{}(v)\right|
}
\leq
\frac{2D(u)D(v)}{D(u)+D(v)}.
\end{align*}
To see that the above statement is true,
observe that $\Nrem{}(u) \leq D(u)$,
$\Nrem{}(v) \leq D(v)$, and that
both the LHS and RHS can be viewed as two resistors
in parallel.
Now it suffices to show the following the algebraic identity:
\begin{align*}
2 x \log{x} + 2 y \log{y} + \frac{2xy}{x + y}
\le
2 \left( x + y \right) \log\left( x + y\right),
\end{align*}
and let $x=D(u)$ and $y=D(v)$.
By symmetry, we can assume $x\le y$ without loss of generality.
Then we get
\begin{align*}
\frac{xy}{x + y}
& \leq \frac{xy}{y}\\
& = y \cdot \frac{x}{y}\\
& \leq y \cdot \log\left(1 + \frac{x}{y} \right),
\end{align*}
where the last inequality follows from $\log(1 + z) \ge z$ when
$z \le 1$.
Plugging this in then gives:
\begin{align*}
2 x \log{x} + 2 y \log{y} + \frac{2xy}{x + y}
& \leq 2 x \log{x} + 2 y \left( \log{y}
+ \log\left(1 + \frac{x}{y} \right) \right)\\
& = 2 x \log{x} + 2 y \log\left( x + y\right)\\
& \leq 2 \left( x + y \right) \log\left( x + y \right). \qedhere
\end{align*}
\end{proof}

\begin{lemma}
  \label{lem:remaining-updates-from-meld}
  Over any fixed sequence of calls to \textsc{Meld}, the expected number of
  updates to the \variable{fill} heaps in remaining vertices (lines~\ref{algline:informvaboutu} and~\ref{algline:informuaboutv}) is bounded by $O(m \log{n})$.
\end{lemma}

\begin{proof}
  By Lemma~\ref{lem:potential-increase}, the number of updates is within a
  constant of the potential increase.
  Since our potential function $\Phi$ is bounded between $0$ and $O(m\log n)$,
  and at no point can it decrease. Hence, the total number of updates is also bounded by $O(m\log n)$.
\end{proof}

Combining the above lemmas gives our main theorem from
Section~\ref{sec:Sketching} on maintaining one copy of the 1-neighborhood sketch.

\begin{proof}(of Theorem~\ref{thm:DataStructureMain})
  Given any graph $G$ and a fixed sequence of vertices for pivoting, we use the
  \textsc{PivotVertex} routine to produce the sequence of graph states
  \begin{align*}
    G = \Gcomp{0}, \Gcomp{1}, \Gcomp{2},\dots,\Gcomp{n} = \emptyset.
  \end{align*}
  Recall that the goal is to maintain $\Rmin(\Nrem{}(w))$ for
  all $w\in \Vcomp{}$ and $\Rmin(\Nfill{}(v))$ for all
  $v\in \Vrem{}$.
  This is achieved by maintaining the two min-heaps, \variable{remaining} and
  \variable{fill}.

  When pivoting a remaining vertex $v$, \textsc{PivotVertex} first makes
  a constant number of updates to $v$'s remaining neighbors,
  which are bounded above by the original degree of $v$. Next, it removes $v$ from the
  \variable{remaining} heaps among $v$'s component vertex neighbors
  (line~\ref{algline:remove-from-component}), which are again at most as many as the original degree of $v$.
  Therefore the total cost of this part of the algorithm is $O(m\log n)$.
  The major chunk of the running time cost is incurred by updates to the
  \variable{fill} heaps in \textsc{InformRemaining}.
  By Lemma~\ref{lem:RemainingUpdatesDueToDeletion} and
  Lemma~\ref{lem:remaining-updates-from-meld}, the number of such updates is bounded by
  $O(m\log n)$.
  As each update is a $O(\log n)$ operation on a heap, the the total running
  time is $O(m\log^2n)$.
  The final step of a meld consists of merging two \variable{remaining}
  heaps.
  Since two min-heaps can be merged in time $O(\log^2 n)$,
  and the number of merges for a pivoted vertex can be bounded by its original degree,
  the cost of this step can be bounded by $O(m\log^2 n)$.
\end{proof}

%% file: Hardness.tex

\section{SETH-Hardness for Computing Minimum Degree Orderings}
\label{sec:Hardness}

Our hardness results for computing the minimum
fill degree and the min-degree ordering are based on
the \emph{strong exponential time hypothesis} (SETH),
which states that for all $\theta > 0$
there exists a~$k$ such that solving $k$-SAT requires
$\Omega(2^{(1-\theta)n})$ time.
Many hardness results based on SETH, including ours, go through the \OV problem
and make use of the following result.

\begin{theorem}[\cite{Williams05}]
\label{thm:HardnessOV}
Assuming SETH,
for any $\theta > 0$, there does not exist an $O(n^{2 - \theta})$ time algorithm
that takes $n$ binary vectors with $\Theta(\log^{2}{n})$ bits
and decides if there is an orthogonal pair.
\end{theorem}

\noindent
We remark that \OV is often stated as deciding if there exists
a pair of orthogonal vectors from two different sets~\cite{Williams2015hardness},
but we can reduce the problem to a single set by appending $[1; 0]$ to all
vectors in the first set and $[0; 1]$ to all vectors in the second set.

The first hardness observation for computing the minimum fill degree 
in a partially eliminated graph is a direct reduction to \OV.
To show this, we construct a bipartite graph that demonstrates how \OV can be
interpreted as deciding if a union of cliques covers the edges of a
clique on the remaining vertices of a partially eliminated graph.

\begin{lemma}
\label{lem:HardnessSingleStep}
Assuming SETH, for any $\theta > 0$,
there does not exist an $O(m^{2 - \theta})$ time algorithm that takes 
as input $G$ with a set of eliminated vertices and computes the minimum fill degree in
$\Gfill{}$.
\end{lemma}

\begin{proof}
Consider an $\OV$ instance with $n$ vectors
$\aa_1, \aa_2, \dots, \aa_n \in \{0,1\}^d$,
and construct a bipartite graph $G = (\Vvec, \Vdim, E)$
such that each vertex in $\Vvec$ corresponds to a vector $\aa_i$
and each vertex in $\Vdim$ uniquely corresponds to a dimension
$1 \leq j \leq d$.
For the edges, we connect vertices $i \in \Vvec$
with $j \in \Vdim$ if and only if $\aa_i(j) = 1$.

Consider the graph state with all of $\Vdim$ eliminated
and all of $\Vvec$ remaining.
We claim that there exists a pair of orthogonal vectors
among $\aa_1, \aa_2,\dots,\aa_n$ if and only if
there exists a remaining vertex $v \in V(G^{+})$
with $\deg(v) < n - 1$.
Let $u, v \in \Vvec$ be any two different vertices,
and let $\aa_u$ and $\aa_v$ be their corresponding vectors.
The vertices $u$ and $v$ are adjacent in $G^{+}$ if and only if there
exists an index $1 \leq j \leq d$
such that $\aa_u(j) = \aa_v(j) = 1$.

Suppose there exists an $O(m^{2 - \theta})$ time algorithm for implicitly
finding the minimum fill degree in a partially eliminated graph, for some
$\theta > 0$.  Then for $d = \Theta(\log^2 n)$ we can use this algorithm to
compute the vertex with minimum fill degree in the bipartite graph described
above in time
\[
  O\left(m^{2-\theta}\right)
  = O\left(\left(n \log^2 n\right)^{2-\theta}\right)
  = O\left(n^{2 - \theta / 2}\right),
\]
which contradicts SETH by Theorem~\ref{thm:HardnessOV}.
\end{proof}

Building on the previous observation, we now show that an exact linear-time
algorithm for computing min-degree elimination orderings is unlikely.  In
particular, our main hardness result is:

\begin{restatable}[]{theorem}{HardnessOrdering}
\label{thm:HardnessOrdering}
Assuming SETH,
for any $\theta > 0$, there does not exist
an $O(m^{4/3 - \theta})$ time algorithm for producing a 
min-degree elimination ordering.
\end{restatable}

The main idea behind our construction is to modify the bipartite graph
in the proof of Lemma~\ref{lem:HardnessSingleStep} in such a way that
a minimum degree ordering has the effect of eliminating
the $d$ vertices in $\Vdim$ before any vertex in~$\Vvec$. This allows us to
use \MinDeg to efficiently solve any instance of \OV.
A limitation of the initial construction is that vertices in
$\Vdim$ can have degree as large as $n$, so requiring that they be
eliminated first is difficult to guarantee.
To overcome this problem, we create a degree hierarchy by splitting 
each vertex in $\Vdim$ into $\Theta(n)$ vertices with degree $O(\sqrt{n})$.
We call this construction a \emph{covering set system} because it maintains all
two-step connections between vertices in $\Vvec$.

\begin{lemma}
\label{lem:CoveringSetSystem}
Given any positive integer $n$,
we can construct in $O(n^{3/2})$ time
a covering set system of the integers $[n] = \{1,2, \dots, n\}$.
This system is collection of subsets $I_1, I_2, \dots, I_k \subseteq [n]$
such that:
\begin{itemize}
  \item The number of subsets $k = O(n)$.
  \item The cardinality $|I_j| \le 10 \sqrt{n}$, for all $1 \leq j \leq k$.
  \item For each $(i_1, i_2) \in [n]^2$ there exists a subset $I_j$ such that
  $i_1, i_2 \in I_j$.
\end{itemize}
\end{lemma}
\noindent
We also pad each vertex in $\Vvec$ with $\Omega(\sqrt{n})$
edges to ensure that it is eliminated after the vertices
introduced by the covering set systems.
We formally describe this construction in Figure~\ref{fig:HardnessOrderingConstruction}.

\begin{figure}[H]

\begin{algbox}
\begin{enumerate}
\item Create one vertex per input vector $\aa_1,\aa_2,\dots,\aa_n$,
and let these vertices be $\Vvec$.
\item For each dimension $j=1$ to $d$:
  \begin{enumerate}
    \item Construct a covering set system for $[n]$.
  \item Create a vertex in $\Vdim$ for each subset in this covering set
  system.
  \item For each vector $\aa_i$ such that $\aa_i(j) = 1$,
    add an edge between its vertex in $\Vvec$ and every vertex corresponding
    to a subset in this covering system that contains $i$.
  \end{enumerate}
\item Introduce $20 \sqrt{n}$ extra vertices called $\Vpad$:
  \begin{enumerate}
  \item Connect all pairs of vertices in $\Vpad$.
  \item Connect each vertex in $\Vpad$ with every vertex in $\Vvec$.
  \end{enumerate}
\end{enumerate}

\end{algbox}

\caption{Construction for reducing \OV to \MinDeg.}

\label{fig:HardnessOrderingConstruction}
\end{figure}

\begin{lemma}
\label{lem:size}
For any \OV instance with $n$ vectors of dimension $d$, let $G$ be the graph
produced by the construction in Figure~\ref{fig:HardnessOrderingConstruction}.
We have $|V| = O(nd)$ and $|E| = O(n^{3/2}d)$.
\end{lemma}
\begin{proof}
The number of vertices in $G$ is
\[
  \left|V\right|
  =
  20 \sqrt{n} + n + d \cdot O\left(n\right)
  = O\left(nd\right).
\]
Similarly, an upper bound on the number of edges in $G$ is
\[
 \left|E\right|
  = \binom{20 \sqrt{n} }{2} + 20 \sqrt{n} \cdot n
  + d \cdot 10\sqrt{n} \cdot O\left(n\right) 
  = 
  O\left(n^{3/2} d\right),
\]
where the terms on the left-hand side of the final equality correspond to edges
contained in $\Vpad$,
the edges between $\Vpad$ and $\Vvec$,
and edges between $\Vvec$ and $\Vdim$, respectively.
\end{proof}

\begin{lemma}
\label{lem:ConstructionElimOrdering}
Consider a graph $G$ constructed from an \OV instance
as described in Figure~\ref{fig:HardnessOrderingConstruction}.
For any min-degree ordering of $G$, the vertices
in $\Vdim$ are the first to be eliminated.
Furthermore, the fill degree of the next vertex to be eliminated is
$\min_{v \in \Vvec} \degfill{}(v)$.
\end{lemma}

\begin{proof}
Let the graph be $G = (V, E)$, such that $V$ is partitioned into
\[
  V = \Vvec \cup \Vdim \cup \Vpad,
\]
as described in Figure~\ref{fig:HardnessOrderingConstruction}.
Initially, for every vertex $\vpad \in \Vpad$ we have
\[
  \deg\left( \vpad \right)
  =
  \left(20 \sqrt{n} - 1\right) + n.
\]
For every vertex $\vvec \in \Vvec$ we have
\[
  \deg\left(\vvec\right)
  =
  20 \sqrt{n} + \left| E\left( \vvec, \Vdim \right) \right|
  \geq
  20 \sqrt{n},
\]
and for every vertex $\vdim \in \Vdim$ we have
\[
  \deg\left(\vdim \right) \le 10 \sqrt{n}.
\]

Pivoting out a vertex in $\Vdim$ does not increase the fill degree of any other
vertex in $\Vdim$ since no two vertices in $\Vdim$ are adjacent.
As these vertices are pivoted, we still maintain
\[
  \degfill{}(v) \ge 20 \sqrt{n},
\]
for all $v \in \Vvec$.
Therefore, the first vertices to be pivoted must be all $v \in \Vdim$.
After all the vertices in $\Vdim$ have been pivoted,
the next vertex has fill degree $\min_{v \in \Vvec} \degfill{}(v)$,
because either a vertex in $\Vvec$ will be eliminated
or all remaining vertices have fill degree $20\sqrt{n} + n - 1$.
\end{proof}

\begin{proof}[Proof of Theorem~\ref{thm:HardnessOrdering}.]
Suppose for some $\theta > 0$ there exists an $O(m^{4/3 - \theta})$ time algorithm
for \MinDeg.
Construct the graph $G=(V,E)$ with covering sets as described in
Figure~\ref{fig:HardnessOrderingConstruction}.
For $d = \Theta(\log^2 n)$,
it follows from Lemma~\ref{lem:size} that
$|V|=O(n \log^2 n)$
and
$|E| = O(n^{3/2} \log^2 n)$.
Therefore, by the assumption, we can obtain
a min-degree ordering of $G$ in time
\[
  O\left(m^{4/3 - \theta}\right)
  = O\left( \left(n^{3/2} \log^2 n\right)^{4/3 - \theta}\right)
  = O\left(n^{2 - \theta}\right).
\]

By Lemma~\ref{lem:ConstructionElimOrdering},
after the first $|\Vdim|$ vertices have been pivoted,
the fill graph $\Gfill{}$ is essentially identical to
the partially eliminated state of the bipartite graph in the proof of Lemma~\ref{lem:HardnessSingleStep}.
We can then compute the fill degree of the next vertex
to be eliminated in $O(m) = O(n^{2 - \theta})$ time by Lemma~\ref{lem:ComputeFill}.
Checking whether the fill degree of this vertex is $20\sqrt{n} + n - 1$
allows us to solve \OV in
$O(n^{2-\theta})$ time, which contradicts SETH.
\end{proof}

All that remains is to efficiently construct the covering set system defined in
Lemma~\ref{lem:CoveringSetSystem}. We can interpret this construction as a way
to cover all the edges of $K_n$ using $O(n)$ $K_{10\sqrt{n}}$ subgraphs.
We note that our construction is closely related to
Steiner systems obtained via finite affine planes
as well as existence results for covering problem with fixed-size
subgraphs~\cite{Chee2013covering,Caro1998covering}.

\begin{proof}[Proof of Lemma~\ref{lem:CoveringSetSystem}.]
We use a simple property of finite fields.
Let $p = \textsc{NextPrime}(\sqrt{n})$, which
we can compute in $O(n)$ 
since $p < 4 \sqrt{n}$ by Bertrand's postulate.
Clearly $[n] \subseteq [p^2]$, so it suffices to find a covering for $[p^2]$.
Map the elements of $[p^2]$ to the coordinates of a $p \times p$ array in the
canonical way so that
$1 \mapsto (0,0), 2 \mapsto(0, 1), \dots, p^2 \mapsto (p-1,p-1)$. 
For all $(a,b) \in \{0,1,\dots,p-1\}^2$,
define
\[
D\left(a,b\right)
=
\left\{
  \left(x,y\right) \in \left\{0,1,\dots,p-1\right\}^2
  : y \equiv ax + b \pmod{p}
\right\}
\]
to be the diagonal subsets of the array,
and define
\[
R\left(a\right)
=
\left\{
  \left(x,y\right) \in \left\{0,1,\dots,p-1\right\}^2 :
  x \equiv a \pmod{p}
\right\}
\]
to be the row subsets of the array.
Let the collection of these subsets be
\[
  S
  =
  \left\{
    D\left(a,b\right) : a,b \in \left\{0,1,\dots,p-1\right\}
  \right\}
  \cup
  \left\{
     R\left(a\right) : a \in \left\{0,1,\dots,p-1\right\}
  \right\}.
\]

The construction clearly satisfies the first two conditions.
Consider any $(a,b) \in [p^2]^2$ and their coordinates
in the array $(x_1,y_1)$ and $(x_2,y_2)$.
If $x_1 = x_2$, then $(x_1,y_1), (x_2,y_2) \in R(x_1)$.
Otherwise, it follows that $(x_1,y_1)$ and $(x_2,y_2)$ are solutions to the line
\[
y
\equiv
  \frac{y_1 - y_2}{x_1 - x_2}
\left( x - x_1 \right) + y_1 \pmod{p},
\]
so the third condition is satisfied.
\end{proof}

%% file: SketchingProofs.tex
\section{Guarantees for Selection-Based Estimators}
\label{sec:SketchingProofs}

In this section we prove Lemma~\ref{lem:minValue},
which states that the reciprocal of the $\lfloor k(1-1/e) \rfloor$-quantile in
$\variable{minimizers}[u]$ can be used to accurately approximate of
$\deg(u)$.
Our proofs follow the same outline as in \cite[Section 7]{Cohen97},
but we consider keys $R(u)$ drawn uniformly from $[0,1)$ instead of
the exponential distribution.
We restate the lemma for convenience.

\minValue*

We start by stating Hoeffding's tail inequality for sums of independent
Bernoulli random variables, and then we give a useful numerical bound that
relates the approximation error and degree of a vertex to the probability that
the quantile $Q(u)$ variable deviates from its expected value.

\begin{lemma}[Hoeffding's inequality]
\label{lem:Hoeffding}
Let $X_1, X_2, \dots, X_n$ be i.i.d.\ Bernoulli random variables such that
$\Pr[X_i = 1] = p$ and $\Pr[X_i = 0] = 1-p$.
Then for any $\delta > 0$ we have the inequalities
\begin{align*}
\Pr \left[ \sum_{i=1}^n X_i \leq (p-\delta)n \right] \leq \exp(-2\delta^2 n), \\
\Pr \left[ \sum_{i=1}^n X_i \geq (p+\delta)n \right] \leq \exp(-2\delta^2 n).
\end{align*}
\end{lemma}

\begin{lemma}
\label{lem:expbounds}
For any $|\epsilon| < 0.1$ and $d \ge 1$, we have
\[
  \exp\left( - 1 + \epsilon - \frac{1}{d+1} \right)
  \leq
  \left( 1 - \frac{1 - \epsilon}{d+1} \right)^{d+1}
  \leq \exp\left( - 1 + \epsilon \right).
\]
\end{lemma}

\begin{proof}
The Maclaurin series for $\log(1 - x)$ is
\[
  \log \left( 1 - x \right)
  =
  -x - \frac{1}{2}x^2 - \frac{1}{3}x^3 - \frac{1}{4}x^4 \dots,
\]
for $-1\le x <1$.
Whenever $|x| \leq 0.1$, we have the inequality
\[
  \left| \frac{1}{3}x + \frac{1}{4}x^2 + \frac{1}{5}x^3 + \dots \right|
  \leq
  \frac{0.1}{3}
  + \frac{0.01}{4}
  + \frac{0.001}{5}
  + \dots
  \leq \frac{1}{2}.
\]
It follows that
\[
-x - x^2
\leq
\log \left( 1 - x \right)
\leq -x.
\]

Applying this inequality when $d+1 \ge 10$ and 
$x=(1-\epsilon)/(d+1) < 0.1$ gives
\[
  -\frac{1 - \epsilon}{d+1} - \frac{1}{(d+1)^2}
\leq
\log \left( 1 - \frac{1 - \epsilon}{d+1} \right)
\leq -\frac{1 - \epsilon}{d+1}.
\]
The result for $d+1 \ge 10$ follows by multiplying the inequalities by $d+1$
and then exponentiating.
Checking the remaining cases numerically completes the proof.
\end{proof}

For convenience, we split the proof of Lemma~\ref{lem:minValue} into two
parts---one for the upper tail inequality and one for the lower tail
inequality.

\begin{lemma}
\label{lem:minValueUB}
Assuming the hypothesis in Lemma~\ref{lem:minValue}, we have
\[
  \Pr\left[Q(u) \ge \frac{1+\epsilon}{\deg(u)+1}\right] \le \frac{1}{n^4}.
\]
\end{lemma}

\begin{proof}
For each sketch $i \in [k]$, we have
\begin{align*}
  \Pr\left[R\left(\minimizer_i(u)\right) \geq \dfrac{1+\epsilon}{\deg(u)+1}\right]
  &= \prod_{v \in N(u)\cup\{u\}} \prob{}{ R_i(v) \geq \dfrac{1+\epsilon}{\deg(u)+1} } \\
  &= \left(1 -  \dfrac{1+\epsilon}{\deg(u)+1} \right)^{\deg(u)+1}.
\end{align*}
	
\noindent
Letting $I_i$ be the indicator variable for the event
$R(\minimizer_i(u)) \geq (1+\epsilon)/(\deg(u)+1)$,
it follows that
\[
  \E[I_i] = \left(1 -  \dfrac{1+\epsilon}{\deg(u)+1} \right)^{\deg(u)+1}
\]
and
\[
  \Pr\left[Q(u) \geq \frac{1 + \epsilon}{\deg(u)+1} \right]
   = \Pr\left[ \sum_{i=1}^k I_i \geq \left\lceil k/e \right\rceil \right].
\]
Since $\E\left[I_i\right] \le \exp(-(1+\epsilon))$,
we let $\delta = 1/e - \E[I_i] > 0$ and use Hoeffding's inequality to show that
\begin{align*}
  \Pr\left[ \sum_{i=1}^k I_i \geq k/e \right]
    &\le \exp\left(-2k\delta^2\right)\\
    &\le \exp\left(-100\log{n} \left(\delta/\epsilon\right)^2\right),
\end{align*}
where the last inequality uses the fact that
$k=50 \left\lceil \log{n} \epsilon^{-2} \right\rceil$.
For any $\epsilon < 1$, we have
\begin{align*}
  \frac{\delta}{\epsilon} &\ge \frac{1}{\epsilon}\left(\frac{1}{e} - \frac{1}{e^{1+\epsilon}}\right) \ge \frac{1}{5}.
\end{align*}
Therefore, it follows that
\begin{align*}
  \Pr\left[Q(u) \ge \frac{1+\epsilon}{\deg(u)+1}\right] \le
  \Pr\left[ \sum_{i=1}^k I_i \geq k/e \right] &\le \frac{1}{n^4},
\end{align*}
as desired.
\end{proof}

\begin{lemma}
\label{lem:minValueLB}
Assuming the hypothesis in Lemma~\ref{lem:minValue}, we have
\[
  \Pr\left[Q(u) \le \frac{1-\epsilon}{\deg(u)+1}\right] \le \frac{1}{n^4}.
\]
\end{lemma}

\begin{proof}
For each sketch $i \in [k]$, we have
\begin{align*}
  \Pr\left[R\left(\minimizer_i(u)\right) \geq \dfrac{1-\epsilon}{\deg(u)+1} \right]
  &= \prod_{v \in N(u)\cup\{u\}} \Pr\left[ R_i(v) \geq \dfrac{1-\epsilon}{\deg(u)+1} \right]\\
  &= \left(1 -  \dfrac{1-\epsilon}{\deg(u)+1} \right)^{\deg(u)+1}.
\end{align*}

\noindent
Letting $J_i$ be the indicator variable for the event
$R(\minimizer_i(u)) \geq (1-\epsilon)/(\deg(u)+1)$, it follows that
\[
  \E\left[J_i\right] = \left(1 -  \dfrac{1-\epsilon}{\deg(u)+1} \right)^{\deg(u)+1}
\]
and
\begin{align*}
  \Pr\left[ Q(u) < \frac{1 - \epsilon}{\deg(u)+1} \right]
     &= \Pr\left[ \sum_{i=1}^{k} J_i \le \left\lceil k/e \right\rceil\right].
\end{align*}

\noindent
Let $\delta = \E[J_i] - 1/e$.
Using Lemma~\ref{lem:expbounds} and the assumption that $\deg(u)+1 > 2\epsilon^{-1}$,
observe that
\begin{align*}
  \delta 
   &\ge \exp\left(-1 + \epsilon - \frac{1}{\deg(u)+1}\right) - 1/e
   > 0.
\end{align*}
Therefore, by Hoeffding's inequality we have
\begin{align*}
  \Pr\left[ \sum_{i=1}^k J_i \leq k/e \right] &\leq \exp(-2k\delta^2)\\
    &\le \exp\left(-100 \log n \left(\delta/\epsilon\right)^{2}\right).
\end{align*}
For any $\epsilon < 1$, using the lower bound for $\delta$ and the
assumption that $\deg(u)+1 > 2 \epsilon^{-1}$ gives
\begin{align*}
  \frac{\delta}{\epsilon} &\ge \frac{1}{e\cdot\epsilon}\left(\exp\left(\epsilon - \frac{1}{\deg(u)+1}\right) - 1\right)
  \ge \frac{e^{\epsilon/2} - 1}{e\cdot\epsilon}
  \ge \frac{1}{2e}.
\end{align*}
Therefore, it follows that
\begin{align*}
  \Pr\left[Q(u) \le \frac{1-\epsilon}{\deg(u)+1}\right] \le
  \Pr\left[ \sum_{i=1}^{k} J_i \le \left\lceil k/e \right\rceil\right]
  \le \frac{1}{n^4},
\end{align*}
which completes the proof.
\end{proof}